\documentclass[a4paper]{article}

\usepackage{a4wide}

\usepackage{amsmath}
\usepackage{amssymb}
\usepackage{amsthm}
\usepackage{array}
\usepackage{graphicx}
\usepackage{calc}
\usepackage{wrapfig}
\usepackage{url}
\usepackage{tikz}
\usetikzlibrary{arrows,automata,shapes,shapes.multipart,decorations.markings,positioning}
\usepackage[all]{xy}
\newcommand{\outputs}{\mathsf{outputs}}

\newcommand{\net}{\mathsf{net}}
\newcommand{\word}{\mathsf{word}}
\newcommand{\sent}[1]{\ensuremath{\mathtt{#1}}}
\newcommand{\FF}{{\cal F}}
\newcommand{\GG}{{\cal G}}
\newcommand{\BB}{{\cal B}}
\newcommand\tuple[1]{\langle #1 \rangle}
\newcommand{\sset}[2]{\left\{~#1  \left|
      \begin{array}{l}#2\end{array}
    \right.     \right\}}

\renewcommand{\vec}{\bar}

\newcommand{\once}{\mathop{at\_most\_one}}

\newcommand{\networkVars}{\mathtt{V}}
\newcommand{\usedVars}{\mathtt{U}}

\newcounter{sncolumncounter}
\newcounter{snrowcounter}

\newcommand{\bee}{\textsf{BEE}}

\newtheorem{theorem}{Theorem}
\newtheorem{lemma}{Lemma}
\newtheorem{corollary}{Corollary}
\newtheorem{conjecture}{Conjecture}

\newtheorem{definition}{Definition}
\newtheorem{remark}{Remark}

\def \nodelabel#1{%
\setcounter{snrowcounter}{1}
 \foreach \i in {#1}{%
   \draw (\sncolwidth*\value{sncolumncounter},\value{snrowcounter}) node[anchor=south]{\i};
   \addtocounter{snrowcounter}{1}
 }
\addtocounter{snrowcounter}{-1}
 \addtocounter{sncolumncounter}{1}
}

\newcommand{\sncolwidth}{0.7}

\def \addcomparator#1#2{%
    \draw (\sncolwidth*\value{sncolumncounter},#1) node[circle,fill=black,minimum size=4pt,inner sep=0pt,outer sep=0pt]{}--(\sncolwidth*\value{sncolumncounter},#2) node[circle,fill=black,minimum size=4pt,inner sep=0pt,outer sep=0pt]{};
}

\newcommand{\addrevcomparator}[3][0.5]{
    \draw[>=stealth,postaction={decorate},decoration={markings,mark=at position #1 with {\arrow[scale=2]{<}}}] (\sncolwidth*\value{sncolumncounter},#2) node[circle,fill=black,minimum size=4pt,inner sep=0pt,outer sep=0pt]{}--(\sncolwidth*\value{sncolumncounter},#3) node[circle,fill=black,minimum size=4pt,inner sep=0pt,outer sep=0pt]{};
}

\def \addhalfcomparator#1#2{%
    \draw (\sncolwidth*\value{sncolumncounter},#1) node[circle,fill=black,minimum size=4pt,inner sep=0pt,outer sep=0pt]{}--(\sncolwidth*\value{sncolumncounter},#2);
}
 
\def \addlayer{%
  \addtocounter{sncolumncounter}{1}
}

\def \nextlayer{%
  \draw [dashed] (\sncolwidth*\value{sncolumncounter}+\sncolwidth,0.6)--(\sncolwidth*\value{sncolumncounter}+\sncolwidth,\value{snrowcounter}+0.6);
  \addtocounter{sncolumncounter}{2}
}

\newcommand{\addarrow}[5][0.75]{
  \draw[>=stealth,postaction={decorate},line width=1.5 pt,decoration={markings,mark=at position #1 with {\arrow[scale=1.2]{>}}}] (\sncolwidth*#2,#3)--(\sncolwidth*#4,#5);
}

\newenvironment{sortingnetwork}[2]
{
  \setcounter{sncolumncounter}{0}
  \setcounter{snrowcounter}{#1}
  \def \sn@fullsize{15}
  \begin{tikzpicture}[scale=#2*0.7]
}
{
  \foreach \i in {1, ..., \value{snrowcounter}}
  {
    \draw (-\sncolwidth,\i)--(\sncolwidth*\value{sncolumncounter}+\sncolwidth,\i);
  }
  \end{tikzpicture}
}
\makeatother

\begin{document}

\title{Optimal-Depth Sorting Networks\thanks{%
    Supported by the Israel Science  Foundation, grant 182/13
    and by the Danish Council for Independent Research, Natural Sciences.
    Computational resources
    provided by an IBM Shared University Award (Israel).}}
\author{Daniel Bundala\\
  \small Department of Computer Science\\
  \small University of Oxford
  \and Michael Codish\\
  \small Department of Computer Science\\
  \small Ben-Gurion University of the Negev
  \and Lu\'\i s Cruz-Filipe\\
  \small Dept.\ Mathematics and Computer Science\\
  \small University of Southern Denmark
  \and Peter Schneider-Kamp\\
  \small Dept.\ Mathematics and Computer Science\\
  \small University of Southern Denmark
  \and Jakub Z\'avodn\'y\\
  \small Department of Computer Science\\
  \small University of Oxford}

\date{}

\maketitle

\begin{abstract}
We solve a 40-year-old open problem on the depth optimality
of sorting networks.
In 1973, Donald E.~Knuth detailed, in Volume~3 of ``\emph{The Art of
  Computer Programming}'', sorting networks of the smallest depth known
at the time for $n \leq 16$ inputs, quoting optimality for $n \leq 8$. 
In 1989, Parberry proved the optimality of the networks with
$9 \leq n \leq 10$ inputs.
In this article, we present a general technique for obtaining such optimality
results, and use it to prove the optimality of the remaining open cases
of $11\leq n\leq 16$ inputs.
We show how to exploit symmetry to construct a small set of
two-layer networks on $n$ inputs such that if there is a sorting
network on $n$ inputs of a given depth, then there is one whose first
layers are in this set.
For each network in the resulting set, we construct a propositional
formula whose satisfiability is necessary for the existence of a
sorting network of a given depth. Using an off-the-shelf SAT solver we
show that the sorting networks listed by Knuth are optimal. For $n
\leq 10$ inputs, our algorithm is orders of magnitude faster than the prior
ones.
\end{abstract}

\section{Introduction}
\label{sec:intro}

General-purpose sorting algorithms are based on comparing and
exchanging pairs of inputs. If the order of these
comparisons is predetermined by the number of inputs to sort and does
not depend on their concrete values, then the algorithm is said to be
data oblivious.  Such algorithms are well-suited for e.g.~parallel
sorting or secure multi-party computations, unlike standard
sorting algorithms, such as QuickSort, MergeSort or HeapSort, where
the order of comparisons performed depends on the input data.
  
Sorting networks are a classical formal model for data-oblivious
algorithms~\cite{Knuth73},
where $n$ inputs are fed into networks of $n$
channels 
connected pairwise by comparators.  Each
comparator takes the two inputs from its two channels, compares them,
and outputs them sorted back to the same two channels.
A set of consecutive
comparators can be viewed as a ``parallel layer'' if no two comparators 
act on
the same channel.  A comparator network is a sorting network if the
output on the $n$ channels is always the sorted sequence of the
inputs.

Ever since sorting networks were introduced, there has been a quest to
find optimal sorting networks: optimal size (minimal number of
comparators), as well as optimal depth (minimal number of layers)
networks.
In their celebrated result, Ajtai, Koml\'{o}s and Szemer\'{e}di
~\cite{AKS} give a construction for sorting networks with
$O(n\log n)$ comparators in $O(\log n)$ parallel levels.
These AKS sorting networks are a classical example of an algorithm
optimal in theory, but highly inefficient in practice.
Although they attain the theoretically optimal $O(n\log n)$ number of
comparisons and $O(\log n)$ depth, the AKS networks are infamous for
the large constants hidden in the big-$O$ notation. On the other hand,
already in~1968, Batcher~\cite{Batcher} gave a simple recursive construction
that, even though it creates networks of depth $O(\log^2 n)$, is
superior to AKS networks for all practical values of $n$.

It is of particular interest to construct optimal sorting networks
(both in size and in depth) for specific small numbers of inputs. Such
networks can be used as building blocks to construct more efficient
networks on larger numbers of inputs, for example by serving as base
cases in recursive constructions such as Batcher's odd-even
construction.  

Already in the fifties and sixties various
constructions appeared for small sorting networks on few inputs. In
the 1973 edition of \emph{``The Art of Computer Programming''}~\cite{Knuth73}
(vol.~3, Section~5.3.4), Knuth detailed the smallest sorting networks
known at the time with $n \leq 16$ inputs.

However, showing their optimality has proved to be extremely
challenging. For $n\leq 8$ inputs, optimality was established by
Knuth and Floyd~\cite{Knuth66} in 1973. No further progress had been
made on the problem until~1989, when Parberry~\cite{Parberry89} showed
that the networks given for $n=9$ and $n=10$ are also
optimal. Parberry obtained this result by implementing an exhaustive
search with pruning based on symmetries in the first two parallel
steps in the sorting networks, and executing the algorithm on a
Cray-2 supercomputer. Despite the great increase in available
computational power in the two and a half decades since, his algorithm
would still not be able to handle the case $n=11$ or bigger. More
recently, there were additional attempts~\cite{MorgensternS11} at
solving the $n=11$ case, but we are not aware of any successful one.

In this paper, some 40 years after the publication of the networks by
Knuth, we finally prove their optimality by settling the remaining open
cases of $11\leq n\leq 16$ inputs.
Our approach combines two methodologies: symmetry breaking and Boolean
satisfiability.

\begin{description}
\item[Symmetry Breaking]
We show how to construct a small set $R_n$ of two-layer networks on
$n$ channels such that: if there is a sorting network on $n$ channels
of a given depth, then there is one whose first two layers are in this
set.  We first show how each two-layer network can be represented by a
graph with isomorphic graphs corresponding to equivalent networks.  By
defining a notion of ``relative strength'' between networks that takes
into account their effects on the inputs, we further restrict the set
of two-layer networks.  We show how to characterize the strongest
networks using context-free grammars, which enables us to construct
the sets $R_n$ for up to $n=40$ inputs within two hours of
computation.  For example, $R_{11}$ consists of 28 networks,
enabling us to solve the optimal-depth problem for $n=11$ in terms of
only 28 independent cases as opposed to over one billion cases of all
two-layer networks on $11$ channels.  Similarly, we show that
$|R_{13}|=117$. 

\item[Boolean Satisfiability]
With the first two layers restricted to a small set, we construct a
family of propositional formulas whose satisfiability is necessary for
the existence of sorting networks of a given size. Using an
off-the-shelf SAT solver we show that all the constructed formulas are
unsatisfiable, and hence we conclude that for $n\leq 16$ inputs the
networks listed in
~\cite{Knuth73} are
indeed optimal.  A similar construction, without restricting the first
two layers, is able to find optimal-depth sorting networks for $n\leq
10$ inputs and prove them optimal, 
thus providing independent confirmation of the previously known
results.
\end{description}

We obtained all our results using an off-the-shelf SAT solver running
under Linux on commodity hardware.  It is noteworthy that our
algorithm required a few seconds to prove the optimality of
networks with $n \leq 10$ inputs, whereas for $n=10$ the algorithm
described in~\cite{Parberry89} was estimated to take hundreds of hours
on a supercomputer, and the algorithm described
in~\cite{MorgensternS11} took more than three weeks on a desktop
computer.

This paper is an extended version of~\cite{DBLP:conf/lata/BundalaZ14}
and~\cite{CCS2014}. The first paper presents the theory and experiments for
calculating optimal sorting networks. In the current
paper we construct even smaller sets of ``non-isomorphic'' two-layer
networks using a much faster algorithm (the construction
in~\cite{DBLP:conf/lata/BundalaZ14} does not scale beyond $n=13$
inputs). This new algorithm is a culmination of the work presented in
the second paper~\cite{CCS2014}. However, that paper deals only with
computing the sets of ``relevant'' two-layer networks, and not with
computing the optimal sorting networks as we do in this paper.

\section{Preliminaries on sorting networks}
\label{sec:backgr}

An example of a comparator network on $4$ channels is shown in
Figure~\ref{fig:network-example}. The figure introduces the graphical
notation used throughout the paper to depict comparator
networks. Channels are indicated as horizontal lines (with channel $4$
at the bottom), comparators are indicated as vertical lines connecting
a pair of channels, and layers are separated by dashed lines.
The figure further shows how the inputs $\tuple{5,2,0,7}$ and
$\tuple{0,1,0,1}$ propagate from left to right through the network.

\begin{figure}[!ht]
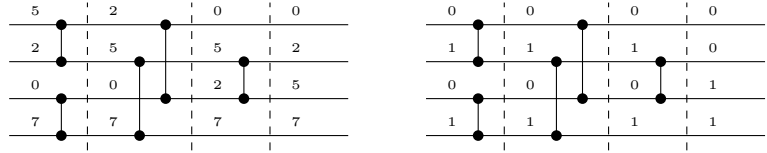

  \centering
\begin{tabular}{cc}
  \begin{sortingnetwork}{4}{0.7}
    \nodelabel{\tiny 7,\tiny 0,\tiny 2,\tiny 5}
    \addcomparator{1}{2}
    \addcomparator{3}{4}
    \nextlayer
    \nodelabel{\tiny 7,\tiny 0,\tiny 5,\tiny 2}
    \addcomparator{1}{3}
    \addlayer
    \addcomparator{2}{4}
    \nextlayer
    \nodelabel{\tiny 7,\tiny 2,\tiny 5,\tiny 0}
    \addcomparator{2}{3}
    \nextlayer
    \nodelabel{\tiny 7,\tiny 5,\tiny 2,\tiny 0}
  \end{sortingnetwork}
&
  \begin{sortingnetwork}{4}{0.7}
    \nodelabel{\tiny 1,\tiny 0,\tiny 1,\tiny 0}
    \addcomparator{1}{2}
    \addcomparator{3}{4}
    \nextlayer
    \nodelabel{\tiny 1,\tiny 0,\tiny 1,\tiny 0}
    \addcomparator{1}{3}
    \addlayer
    \addcomparator{2}{4}
    \nextlayer
    \nodelabel{\tiny 1,\tiny 0,\tiny 1,\tiny 0}
    \addcomparator{2}{3}
    \nextlayer
    \nodelabel{\tiny 1,\tiny 1,\tiny 0,\tiny 0}
  \end{sortingnetwork}
\end{tabular}

\caption{The comparator network
  $\{(1,2),(3,4)\};\{(1,3),(2,4)\};\{(2,3)\}$ with $4$ channels, $5$
  comparators, and depth~$3$. On the left, the input~$\tuple{5,2,0,7}$
  propagates through the network to give the output~$\tuple{0,2,5,7}$;
  on the right, the input~$\tuple{0,1,0,1}$
  propagates through the network to give the output~$\tuple{0,0,1,1}$.}
  \label{fig:network-example}
\end{figure}

Formally, a \emph{comparator network} $C$ with $n$ \emph{channels} and
\emph{depth} $d$ is a sequence $C = L_1;\ldots;L_d$ of $d$
\emph{layers}. Each layer $L_k$ is a set of comparators $(i,j)$,
joining the channels $i$ and $j$, with $1\leq i<j\leq n$.  In every layer $L_k$,
each channel $i$ is used by at most one comparator,
i.e.~\[\left|\sset{j}{(i,j) \in L_k \vee (j,i) \in L_k}\right| \leq 1\]
for each $i$.
The \emph{size} of $C$ is the total number of comparators in all its
layers. Given comparator networks $C_1$ and $C_2$, let $C_1;C_2$ denote
the comparator network obtained by concatenating the layers of $C_1$
and $C_2$. If $C_1$ has $m$ layers, then it is an \emph{$m$-layer prefix}
of $C_1;C_2$.

An input to $C$ is a sequence $\vec x = \tuple{x_1,\ldots,x_n}$ of
numbers. The input is propagated through the network, each comparator
$(i,j)$ outputting the smaller of its inputs on channel~$i$ and the
larger of the inputs on channel~$j$.

Denote by $C(\vec x,k,i)$ the value of channel $1 \leq i \leq n$ at
layer $0\leq k\leq d$ given input $\vec x$. Then we take $C(\vec
x,0,i) = x_i$ (input), and for $0\leq k< d$ we define:
\[
C(\vec x,k+1,i) = \begin{cases}
\min(C(\vec x,k,i), C(\vec x,k,j)) & \textrm{if $(i,j) \in L_{k+1}$} \\
\max(C(\vec x,k,i), C(\vec x,k,j)) & \textrm{if $(j,i) \in L_{k+1}$} \\
C(\vec x,k,i) & \textrm{otherwise.} \\
\end{cases}
\]
The \emph{output} of $C$ on $\vec x$ is the sequence $C(\vec x) =
\tuple{C(\vec x,d,1), C(\vec x,d,2), \ldots, C(\vec x,d,n)}$. 
A comparator network is called a \emph{sorting network} if the output
$C(\vec x)$ is sorted (ascendingly) for all input sequences $\vec
x$. The comparator network depicted in
Figure~\ref{fig:network-example} is a sorting network. The figure
further indicates the values $C(\vec x, k,i)$ on each channel $i$
after each layer $k$ for $\vec x=\tuple{5,2,0,7}$ (on the left) and
for $\vec x= \tuple{0,0,1,1}$ (on the right).

Given sufficient parallel computational power (e.g., assuming the
networks are directly implemented in hardware), independent
comparators can be evaluated in parallel, and hence the depth of a
sorting network corresponds to the number of parallel steps needed to
sort $n$ inputs. Thus, given number of channels $n$, we focus on
finding sorting networks of minimal depth.\footnote{In general, it is possible
  to construct a sorting network with fewer comparators, albeit larger
  depth, than the one that achieves $T(n)$. We refer the readers
  interested in the minimum number of comparators needed to sort $n$
  channels to~\cite{Knuth73} and~\cite{CCFS2014} where the optimal
  values are presented for $n \leq 8$ and $n = 9, 10$, respectively.}
We denote the smallest
depth of a sorting network on $n$ channels by~$T(n)$.

Prior to this paper, the precise values of $T(n)$ were known only for
$n\leq 10$: the values for $n\leq 8$ are given in~\cite{Knuth73}, and
those for $n=9,10$ are reported by Parberry in~\cite{Parberry91}.
These, and the best previously known bounds for $n\leq 16$, are summarized in
Table~\ref{tab:previousoptimal}.

\begin{table}[!ht]
  \centering
    \begin{tabular}{|r|c|c|c|c|c|c|c|c|}
    \hline 
    $n$ & 1 & 2 & 3,4 & 5,6 & 7,8 & 9,10 & 11,12 & 13,14,15,16 \\
    \hline
    $T(n)\leq$ & 0 & 1 & 3 & 5 & 6 & 7 & 8 & 9 \\
    \hline
    $T(n)\geq$ & 0 & 1 & 3 & 5 & 6 & 7 & 7 & 7\\
    \hline
    \end{tabular}
  \caption{The best previously known upper and lower bounds for $T(n)$.}
  \label{tab:previousoptimal}
\end{table}

In this paper, we prove optimality
of the upper bounds of $T(n)$ for $11\leq n\leq 16$. For example, we
show that $T(11) \geq 8$. To prove such a result, we have to
establish that none of the $7$-layer, $11$-channel comparator networks
is a sorting network.

Each comparator joins two distinct channels, and hence one
can view each layer of an $n$-channel comparator network as a matching
on $n$ elements~\cite{Parberry91}. It turns out that there are
$35{,}696$ matchings on $11$ elements.\footnote{Sequence
  \texttt{A000085} of the the On-Line Encyclopedia of Integer
  Sequences, published electronically at \url{http://oeis.org}.}
So,
to establish the lower bound $T(11) \geq 8$, we have to show that none
of the $35{,}696^7 \geq 10^{31}$ comparator networks on $11$~channels
with $7$~layers is a sorting network. Similarly, establishing that $T(13)
\geq 9$ requires showing that none of the $568{,}504^8 \geq 10^{46}$
comparator networks on $13$~channels with $8$~layers
is a sorting network.  These
numbers immediately make any form of exhaustive search infeasible.

In the first part of this paper, we show how to reduce the size of
this search space. Then, in Section~\ref{sec:encode}, we reduce the
existence of a sorting network to the problem of propositional
satisfiability.

In full, our algorithm to determine whether a sorting network of a
given depth exists consists of four phases.  In the first phase, we
extend the approach introduced by Parberry in~\cite{Parberry91}, and
partition of the set of two-layer comparator networks into
equivalence classes. We then select representatives of some of these
equivalence classes so that, if there is a sorting network of the given
size, then there is one with the first two layers equal to some
chosen representative.
In the next phase, we reduce the existence of a sorting network
beginning with one of the calculated representatives to satisfiability of a
corresponding propositional formula. Finally, we determine the
satisfiability of the obtained formulas using a SAT solver.

On the face of it, to determine whether a given $n$-channel
candidate comparator network is a sorting one it seems necessary to
try all possible permutations of $\{1,\ldots,n\}$ as inputs. The
following classical result states that it suffices to consider only
Boolean inputs, i.e.~sequences of $0$ and $1$s.
This reduces the size of the set of inputs from $n!$
permutations to $2^n$ Boolean inputs.

\begin{lemma}[\rm The zero-one principle~\cite{Knuth73}]
\label{lemma:zeroone}
A comparator network $C$ is a sorting network if and only if $C$ sorts
all Boolean inputs.
\end{lemma}

In the remainder of the paper, we consider only comparator networks with
Boolean inputs.
We will write $\BB^n=\{0,1\}^n$ to denote the set of Boolean inputs, and,
given a comparator network $C$, we define
$\outputs(C)=\sset{C(\vec x)}{\vec x\in\BB^n}$. Hence, $C$ is
a sorting network if and only if all elements of $\outputs(C)$ are
sorted (in ascending order).

\section{Equivalence of comparator networks}

A first step in reducing the search space of all comparator networks
can already be found in the work of Parberry~\cite{Parberry91}.  A
layer on $n$ channels is called \emph{maximal} if it contains
$\left\lfloor\frac{n}{2}\right\rfloor$ comparators, i.e., no further
comparators can be added to the layer.
\begin{lemma}[Parberry~\cite{Parberry91}]
  \label{lem:parberry}
  Let $L$ be any maximal layer on $n$~channels.  If there is a sorting
  network on $n$~channels with depth~$d$, then there is one whose
  first layer is~$L$.
\end{lemma}

This lemma implies that, when searching for an optimal-depth sorting network,
the first layer can be fixed to any maximal first layer,
effectively reducing the problem by one layer. In this paper we
consider the following choice of the first layer of $n$-channel
sorting networks:
\begin{align*}
  F_n &= \sset{(2i-1,2i)}{1\leq i\leq 
            \left\lfloor\frac{n}{2}\right\rfloor}
\end{align*}
See for example the network of Figure~\ref{fig:network-example},
whose first layer is~$F_4$.

Assuming the first layer has been fixed to some maximal layer,
Parberry also states~\cite{Parberry91} that one need not consider
second layers which are identical modulo permutations of channels.
Let $\pi$ be a permutation on $\{1,\ldots,n\}$.  For a comparator
$(i,j)$ we define $\pi((i,j)) = (\pi(i),\pi(j))$. Then for a layer
$L$, we define \[\pi(L)=\sset{(\pi(i),\pi(j))}{(i,j)\in L}.\] If $i<j$
for each $(i,j)\in \pi(L)$, then $\pi(L)$ is also a layer, otherwise
we say that $\pi(L)$ is a \emph{generalized layer}.  The definition
naturally extends to a network $C=L_1;\cdots;L_k$ by applying the
permutation on each layer independently: $\pi(C) =
\pi(L_1);\cdots;\pi(L_k)$, yielding a \emph{generalized comparator
  network}.\label{label:permutation}
It is well known (see e.g.~Exercise~5.3.4.16 of~\cite{Knuth73})
that a generalized sorting network can always be ``untangled'' into
a standard sorting network of the same dimensions (see Figure~\ref{fig:equiv-nets} for an example). Furthermore, this
operation preserves the ``standard prefix'', i.e., the longest
prefix of the network that does not have generalized layers. Formally, if $E$ is a comparator network and $F$ is a generalized comparator network such that $E;F$ is a generalized sorting network, then there is a comparator network $F'$ of the same depth as $F$ so that $E;F'$ is a sorting network. 

\begin{figure}[hb]
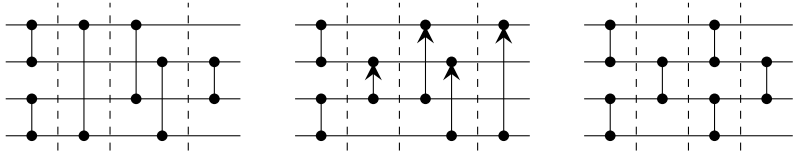

  \centering
  \begin{sortingnetwork}{4}{0.7}
    \addcomparator12
    \addcomparator34
    \nextlayer
    \addcomparator14
    \nextlayer
    \addcomparator24
    \addlayer
    \addcomparator13
    \nextlayer
    \addcomparator23
  \end{sortingnetwork}
  \begin{sortingnetwork}{4}{0.7}
    \addcomparator12
    \addcomparator34
    \nextlayer
    \addrevcomparator32
    \nextlayer
    \addrevcomparator[0.25]42
    \addlayer
    \addrevcomparator[0.25]31
    \nextlayer
    \addrevcomparator[0.17]41
  \end{sortingnetwork}
  \begin{sortingnetwork}{4}{0.7}
    \addcomparator12
    \addcomparator34
    \nextlayer
    \addcomparator23
    \nextlayer
    \addcomparator12
    \addcomparator34
    \nextlayer
    \addcomparator23
  \end{sortingnetwork}
  \caption{The $4$-channel networks to the left and to the right are
    equivalent via the permutation $(1\,3)(2\,4)$. The middle network
    is the generalized comparator network obtained by applying the
    permutation to the left comparator network. The right network can
    then be obtained by untangling the middle network.}
  \label{fig:equiv-nets}
\end{figure}

Taking into account that the first layer is fixed
(Lemma~\ref{lem:parberry}), Parberry~\cite{Parberry91}
considered permutations that leave the first layer intact.

\begin{lemma}[\rm Theorem~5.4 of~\cite{Parberry91}]
\label{lem:parberryLevel2}
Let $L_1$ and $L_2$ be layers on $n$ channels.
Let $\pi$ be a
permutation such that $L_1$ is maximal, $\pi(L_1)=L_1$,
and $\pi(L_2)$ is a layer.
Then there is a depth-$d$ sorting network of the
form $L_1;L_2;C$ if and only if there is one of the form
$L_1;\pi(L_2);C'$.
\end{lemma}

The proof of depth optimality for sorting networks with~$9$ channels
described in~\cite{Parberry91} is based on the application of
Lemma~\ref{lem:parberryLevel2} together with a brute force algorithm
that first partitions the set of two-layer networks with a fixed
maximal first layer into equivalence classes modulo permutations
that fix the first layer. The ``small'' number of equivalence classes
for $n\leq 10$ channels is computed in this way and reported in~\cite{Parberry91}.
However, when partitioning networks into equivalence classes using a
brute-force approach, one must consider the rapidly increasing number
of permutations, and this approach does not scale as the number of
channels grows. Furthermore, even if these equivalence classes were
given, the search algorithm described in \cite{Parberry91} does not
scale for larger numbers of channels.

The main theme of the first half of this paper is a better computation
and exploitation of symmetries in the first two levels of comparator
networks. Using the terminology of the definition below, we aim to
find as small as possible \emph{complete sets of filters}. For
example, for $n=16$ we reduce the number of second layers that must be
considered from $46{,}206{,}736$ to only $211$.

\begin{definition}
\label{def:complete}
A set $\FF$ of comparator networks on $n$ channels is a
complete set of filters for the optimal-depth sorting network problem
if  there exists an optimal-depth sorting network
on $n$ channels  of the form $C;C'$ for
some $C\in\FF$.
\end{definition}

We now introduce a notion of equivalence of comparator networks that
is stronger than the one considered by Parberry~\cite{Parberry91}.
Let $C$ be a comparator network on $n$ channels.  The graph
representation of $C$ is a directed and labeled graph $\GG(C)=(V,E)$,
where each node in $V$ corresponds to a comparator in $C$ and
$E\subseteq V\times \{\sent{1},\sent{2}\}\times V$.  Let $c(v)$ denote
the comparator corresponding to a node $v$.
Then $(u,\sent 1,v)\in E$ if the minimum output of $c(u)$ is an input of $c(v)$,
and $(u,\sent 2,v)\in E$ if the maximum output of $c(u)$ is an input of $c(v)$.
Note that the number of channels cannot be
inferred from this representation, as channels that are unused
are not represented.

Figure~\ref{fig:graphs} illustrates the graph representations of the
left and right networks from Figure~\ref{fig:equiv-nets}, where the
comparators are labeled alphabetically in order of occurrence
(left-to-right, top-down). Note that these two graphs can be seen to be
isomorphic by mapping the vertices as $a\mapsto b'$, $b\mapsto a'$,
$c\mapsto c'$, $d\mapsto d'$, $e\mapsto e'$ and $f\mapsto f'$.

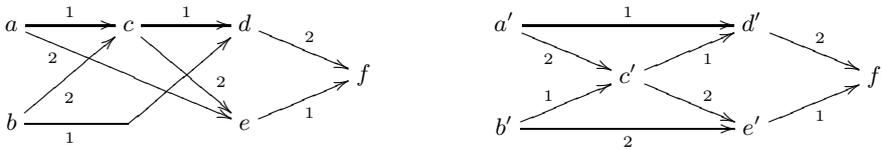
\begin{figure}[htb]
  \centering
\small
$\xymatrix@R-2em@C+1em{
  a \ar[r]^{1} \ar[rrdd]_(.2){2}
  & c \ar[r]^{1} \ar[rdd]^(.7){2}
  & d \ar[rd]^{2}
  \\
  &&& f
  \\
  b \ar[ruu]_(.4){2} \ar'[r]+0_{1}[rruu]
  && e \ar[ur]_{1}
}$
\qquad\qquad
$
\xymatrix@R-2em@C+1em{
  a' \ar[rr]^{1} \ar[dr]_(.4){2}
  &&
  d' \ar[rd]^{2}
  \\
  & c' \ar[ru]_(.6){1} \ar[rd]^(.6){2}
  && f'
  \\
  b' \ar[ur]^(.4){1} \ar[rr]_{2}
  &&
  e' \ar[ur]_{1}
}$
\caption{Graph representations of the two networks in
  Figure~\ref{fig:equiv-nets}.}
  \label{fig:graphs}
\end{figure}

Clearly, graphs representing comparator networks are acyclic, and the
degrees of their vertices are bounded by~$4$.
The strong relationship between equivalence of comparator networks and
isomorphism of their corresponding graphs is reflected in 
Lemma~\ref{lem:isomorphism} below,
which implies that the comparator network equivalence problem is
polynomially reduced to the bounded-valence graph isomorphism problem.

\begin{definition}
  Let $C_1$ and $C_2$ be $n$-channel comparator networks. Then we
  write $C_1 \approx C_2$ if $\GG(C_1)$ and $\GG(C_2)$ are isomorphic.
\end{definition}

\begin{lemma}[\rm Choi \&\ Moon~\cite{DBLP:conf/gecco/ChoiM02}, Proposition~2]
  \label{lem:isomorphism}
  Let $C = C_1;D$ and $C_2$ be $n$-channel comparator networks such
  that $C_1$ and $C_2$ have the same depth and $ \GG(C_1) \approx
  \GG(C_2)$. If $C$ is a sorting network, then there is an $n$-channel
  comparator network $D'$ of the same depth as $D$ such that $C_2;D'$
  is a sorting network.
\end{lemma}

The graph isomorphism problem is one of a very small number of problems
belonging to NP that are neither known to be solvable in polynomial time nor
known to be NP-complete.
However, it is known that the isomorphism of graphs of bounded valence can be
tested in polynomial time~\cite{DBLP:journals/jcss/Luks82}, so the comparator
network equivalence problem can be efficiently solved.

\section{Complete sets of two-layer filters}
\label{sec:prefix}

Recall that our goal is for given $n$ to compute as small as possible
complete sets of filters consisting of two-layer networks on $n$
channels. We now show how to exploit the graph representation to
compute such a set.

\subsection{A Symbolic Representation}

The obvious approach for finding all two-layer prefixes modulo
symmetry is to generate all two-layer networks, and then apply graph
isomorphism to find canonical representatives of the equivalence
classes.  We evaluated this approach using the popular graph
isomorphism tool \verb!nauty!~\cite{DBLP:journals/jsc/McKayP14}, but
found that the exponential growth in the number of two-layer prefixes
prevents this approach from scaling. Therefore, we opted for a
symbolic representation of these graphs that captures isomorphism.

For the special case of two-layer networks, the vertices in the
resulting graphs always have degree~$1$ or~$2$.
Therefore, they always consist of sets of ``sticks'' and ``cycles'',
and they are completely characterized by the maximal-length simple
paths they contain.  Moreover, this representation can be determined
directly from the network, as illustrated in Figure~\ref{ex:words}.

It is useful to adopt the following terminology on channels.  A
channel in a comparator network is called a \emph{min-channel}
(respectively, a \emph{max-channel}) if it is the smaller (resp.\ larger)
channel in some comparator of the first layer. We will also
occasionally refer to a min- or max- channel at a layer~$d$ with the
obvious meaning. A channel of a comparator network is called a
\emph{free} channel if it is not used in the first layer.

\begin{definition}\label{def:path}
  A \emph{path} in a two-layer network $C$ is a sequence
  $\tuple{p_1p_2\ldots p_k}$ of distinct channels such that each pair
  of consecutive channels is connected by a comparator in~$C$.

  The \emph{word} corresponding to $\tuple{p_1p_2\ldots p_k}$ is
  $\tuple{w_1w_2\ldots w_k}$, where:
  \[w_i=\begin{cases}
  \sent{0} & \mbox{if $p_i$ is the free channel} \\
  \sent{1} & \mbox{if $p_i$ is a min-channel} \\
  \sent{2} & \mbox{if $p_i$ is a max-channel}
  \end{cases}\]
\end{definition}

A path is maximal if it is a simple path (with no repeated nodes) that
cannot be extended (in either direction).
A network is connected if its graph representation is connected.

\begin{definition}
  \label{defn:word}
  Let $C$ be a connected two-layer network on $n$ channels.  Then
  $\word(C)$ is defined as follows, where there are three kinds of words.
  \begin{description}
  \item[Head-word.] If $n$ is odd, then\/ $\word(C)$ is the word
    corresponding to the maximal path in $C$ starting with
    the (unique) free channel. 
  \item[Stick-word.] If $n$ is even and $C$ has two channels not used
    in layer~$2$, then $\word(C)$ is the lexicographically smallest of
    the words corresponding to the two maximal paths in $C$ starting
    with one of these unused channels (which are reverse to one another).
  \item[Cycle-word.] If $n$ is even and all channels are used by a
    comparator in layer $2$, then $\word(C)$ is
    the lexicographically
    smallest word corresponding to a maximal path in $C$ that
    begins with two channels connected in layer~$1$.
  \end{description}
\end{definition}

The set of all possible words (not necessarily minimal with respect to
lexicographic ordering) can be described by the following BNF-style grammar.
\begin{align}
  \label{eq:word}
  \mathsf{Word} &::= \mathsf{Head} \mid \mathsf{Stick} \mid \mathsf{Cycle}
  & \mathsf{Stick} &::= (\sent{12}+\sent{21})^+ \\
  \nonumber
  \mathsf{Head} &::= \sent{0}(\sent{12}+\sent{21})^\ast
  & \mathsf{Cycle} &::= \sent{12}(\sent{12}+\sent{21})^+
\end{align}
To avoid ambiguity, we annotate each word with a tag from the set
$\{\sent h,\sent s,\sent c\}$
to indicate whether it is a Head-, Stick- or Cycle-word, respectively.
In Figure~\ref{ex:words}, the three two-layer networks $a$--$c$ lead to the
generation of a word of each kind resulting from the paths shown in $a'$--$c'$,
respectively.

\begin{figure}
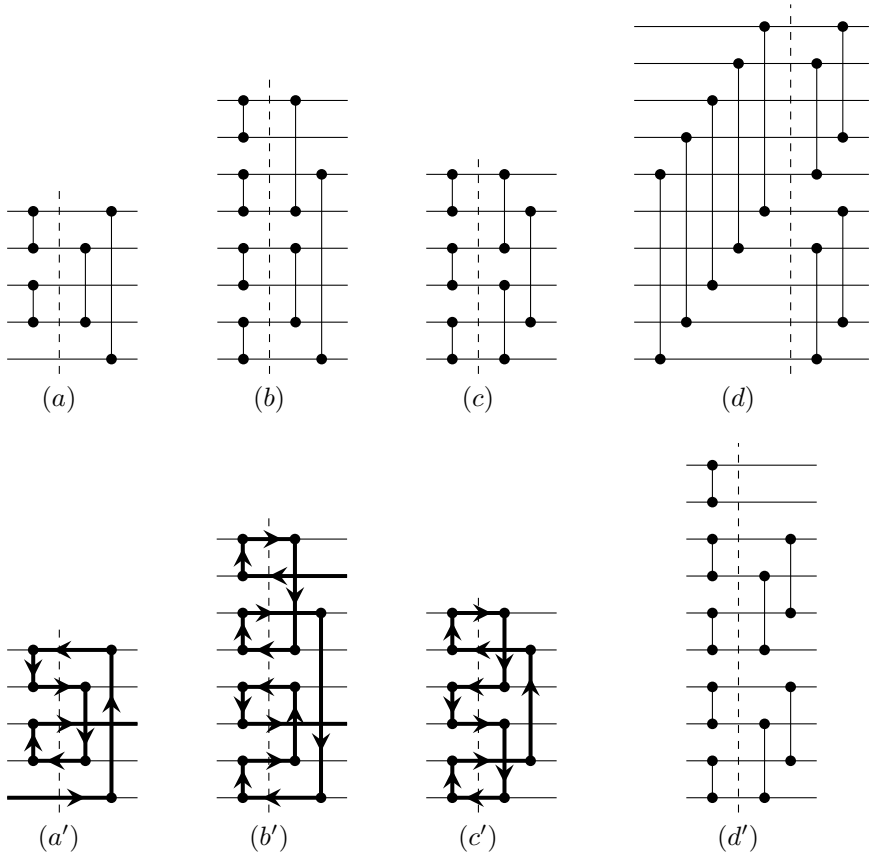

\begin{tabular}{cccc}
\begin{sortingnetwork}{5}{0.7}
    \addcomparator{5}{4}
    \addcomparator{3}{2}
    \nextlayer
    \addcomparator{2}{4}
    \addlayer
    \addcomparator{1}{5}
\end{sortingnetwork}
&
\begin{sortingnetwork}{8}{0.7}
    \addcomparator{7}{8}
    \addcomparator{5}{6}
    \addcomparator{3}{4}
    \addcomparator{1}{2}
    \nextlayer
    \addcomparator{5}{8}
    \addcomparator{2}{4}
    \addlayer
    \addcomparator{1}{6}
\end{sortingnetwork}
&
\begin{sortingnetwork}{6}{0.7}
    \addcomparator{5}{6}
    \addcomparator{3}{4}
    \addcomparator{1}{2}
    \nextlayer
    \addcomparator{4}{6}
    \addcomparator{1}{3}
    \addlayer
    \addcomparator{2}{5}
\end{sortingnetwork}
&
\begin{sortingnetwork}{10}{0.7}
  \addcomparator16
  \addlayer
  \addcomparator27
  \addlayer
  \addcomparator38
  \addlayer
  \addcomparator49
  \addlayer
  \addcomparator5{10}
  \nextlayer
  \addcomparator14
  \addcomparator69
  \addlayer
  \addcomparator25
  \addcomparator7{10}
\end{sortingnetwork}
\\
$(a)$ & $(b)$ & $(c)$ & $(d)$ \\[1em]
\begin{sortingnetwork}{5}{0.7}
    \addcomparator{5}{4}
    \addcomparator{3}{2}
    \nextlayer
    \addcomparator{2}{4}
    \addlayer
    \addcomparator{1}{5}
    \addarrow{-1}131
    \addarrow[0.7]3135
    \addarrow3505
    \addarrow0504
    \addarrow0424
    \addarrow[0.8]2422
    \addarrow2202
    \addarrow0203
    \addarrow[0.4]0343
\end{sortingnetwork}
&
\begin{sortingnetwork}{8}{0.7}
    \addcomparator{7}{8}
    \addcomparator{5}{6}
    \addcomparator{3}{4}
    \addcomparator{1}{2}
    \nextlayer
    \addcomparator{5}{8}
    \addcomparator{2}{4}
    \addlayer
    \addcomparator{1}{6}
    \addarrow4707
    \addarrow0708
    \addarrow0828
    \addarrow[0.6]2825
    \addarrow2505
    \addarrow0506
    \addarrow[0.4]0636
    \addarrow3631
    \addarrow3101
    \addarrow0102
    \addarrow[0.8]0222
    \addarrow2224
    \addarrow2404
    \addarrow0403
    \addarrow[0.4]0343
\end{sortingnetwork}
&
\begin{sortingnetwork}{6}{0.7}
    \addcomparator{5}{6}
    \addcomparator{3}{4}
    \addcomparator{1}{2}
    \nextlayer
    \addcomparator{4}{6}
    \addcomparator{1}{3}
    \addlayer
    \addcomparator{2}{5}
    \addarrow2101
    \addarrow0102
    \addarrow[0.5]0232
    \addarrow3235
    \addarrow3505
    \addarrow0506
    \addarrow0626
    \addarrow[0.8]2624
    \addarrow2404
    \addarrow0403
    \addarrow0323
    \addarrow[0.8]2321
\end{sortingnetwork}
&
\begin{sortingnetwork}{10}{0.7}
  \addcomparator12
  \addcomparator34
  \addcomparator56
  \addcomparator78
  \addcomparator9{10}
  \nextlayer
  \addcomparator13
  \addcomparator57
  \addlayer
  \addcomparator24
  \addcomparator68
\end{sortingnetwork}
\\
$(a')$ & $(b')$ & $(c')$ & $(d')$
\end{tabular}

\caption{
  Networks and paths.
  Networks~($a$--$c$) correspond to the three cases in Definition~\ref{defn:word},
  and ($a'$--$c'$) depict the corresponding maximal paths.
  Path~$(a')$ starts from the free channel,
  with corresponding word $\sent{01221_h}$.
  Path~$(b')$ starts from a free channel at layer~$2$, corresponding to the
  word~$\sent{21212112_s}$; the reverse path corresponds to the smaller word
  $\sent{21121212_s}$, which represents network~$(b)$.
  For the cycle in~$(c')$, the smallest word is $\sent{121221_c}$,
  obtained on the reverse path starting from channel~$1$.\newline
  Network~$(d)$ consists of three connected components (the sets of channels
  $\{1,4,6,9\}$, $\{2,5,7,10\}$ and $\{3,8\}$), corresponding to
  the first two layers of the $10$-channel
  sorting network from Figure~49 of~\cite{Knuth73}.
  The first two components contain cycles represented by $\sent{1221_c}$,
  and the third yields the Stick-word $\sent{12_s}$.
  The network is thus represented by the sentence $\sent{12_s;1221_c;1221_c}$.
  In turn, this sentence generates the equivalent network~$(d')$.
}
\label{ex:words}
\end{figure}

\begin{definition}
  \label{defn:sentence}
  The \emph{word representation} of a two-layer comparator network $C$,
  $\word(C)$, is the multi-set containing $\word(C')$ for each connected
  component $C'$ of $C$; we will denote this set by the ``sentence''
  $w_1;w_2;\ldots;w_k$, where the words are in lexicographic order (including
  their tags).
\end{definition}
In particular, a connected network will be represented by a sentence with only
one word, so there is no ambiguity in the notation $\word(C)$.  The
requirement that the first layer is maximal corresponds to the requirement
that the multi-set $\word(C)$ has at most one Head-word.
Figure~\ref{ex:words}$(d)$ illustrates the case of a multi-component
two-layer network.

Conversely, a word $w$ defines a two-layer network as follows.

\begin{definition}
  \label{defn:net}
  Let $w$ be a word in the language of Equation~\eqref{eq:word}, and $n=|w|$.
  The two-layer network $\net(w)$ has first layer~$F_n$
  and second layer defined as follows.
  \begin{enumerate}
  \item If $w$ is a Stick-word or a Cycle-word, ignore the first character;
    then, for $k=0,\ldots,\left\lfloor\frac n2\right\rfloor-1$, take the next
    two
    characters $xy$ of $w$ and add a second-layer comparator between channels
    $2k+x$ and $2(k+1)+y$.
    Ignore the last character;
    if $w$ is a Cycle-word, connect the two remaining channels at the end.
  \item If $w$ is a Head-word, proceed as above but start by connecting the free
    channel to the channel indicated by the second character.
  \end{enumerate}
\end{definition}

To generate a network from a sentence, we simply generate the networks for each
word in the sentence and compose them in the same order.
Figure~\ref{ex:words}$(d')$ illustrates this construction.

The following lemma shows that abstracting networks to words captures
network equivalence.

\begin{lemma}
  \label{lem:equiv}
  Let $C$ and $C'$ be two-layer comparator networks on $n$ channels. Then $C \approx C'$ if and only if $\word(C)=\word(C')$.
\end{lemma}
\begin{proof}
  The ``if'' part follows from the observation that, for
  two-layer networks, $C \approx C'$ means that there is a permutation
  $\pi$ such that $C'$ equals $\pi(C)$ with possibly some comparators
  reversed in the second layer.
  Thus, any path obtained in $C$ beginning at channel $j$ can be obtained in
  $C'$ by beginning at channel $\pi(j)$, and vice versa.
  The ``only if'' part is straightforward.
\end{proof}

\begin{remark}
In algebraic terms, the functions $\word$ and $\net$ form an adjunction between
the preorders of words (with lexicographic ordering) and two-layer comparator
networks (with equivalence).
The function $\word$ can be seen as a ``forgetful'' functor that forgets the
specific order of channels in a net, whereas $\net$ generates the ``free'' network
with first layer $F_n$ from a given word.
Furthermore, $\word$ always returns the minimum element in the fiber
$\net^{-1}(w)$, whence lexicographic minimal words can be used to
characterize equivalent networks.
\end{remark}
\begin{remark}
Note that Definition~\ref{defn:net} can be adapted to any choice of
a maximal first layer.
\end{remark}

Definition~\ref{defn:net} fixes the first layer of a
two-layer comparator network $\net(w)$ to $F_n$. Denote the set of all
possible second layers (in a two-layer network whose first layer
is $F_n$) by $G_n$. We denote the equivalence classes of the two-layer
networks whose first layer is $F_n$ and whose second layer is a member of
$G_n$ by $R(G_n)$. We view  $R(G_n)$ as a set of representatives of
the equivalence classes. So  $R(G_n)$ is viewed as a maximal set of
non-equivalent networks.

\begin{theorem}
  For any $n\geq 3$, the set $R(G_n)$ of two-layer comparator
  networks is a complete set of filters for the optimal-depth sorting
  network problem.
\end{theorem}

As a consequence of Lemma~\ref{lem:equiv},
$R(G_n)$ can be constructed simply by generating all multi-sets of
words with at most one Head-word yielding exactly $n$ channels.  This
procedure has been implemented straightforwardly in Prolog, resulting
in the values in Table~\ref{tab:Gn}.

\begin{table}
  \[\begin{array}{c|r|r|r|r|r|r|r|r|r|r|r|r}
  n & \multicolumn1{c|}{3}
  & \multicolumn1{c|}{4}
  & \multicolumn1{c|}{5}
  & \multicolumn1{c|}{6}
  & \multicolumn1{c|}{7}
  & \multicolumn1{c|}{8}
  & \multicolumn1{c|}{9}
  & \multicolumn1{c|}{10}
  & \multicolumn1{c|}{11}
  & \multicolumn1{c|}{12}
  & \multicolumn1{c|}{13}
  & \multicolumn1{c}{14}  \\ \hline
  |G_n| & 4 & 10 & 26 & 76 & 232 & 764 & 2{,}620 & 9{,}496 & 35{,}696 & 140{,}152 & 568{,}504 & 2{,}390{,}480 \\
  |R(G_n)| & 4 & 8 & 16 & 20 & 52 & 61 & 165 & 152 & 482 & 414 & 1{,}378 & 1{,}024 \\ \hline
  \end{array}\]
  \[\begin{array}{c|r|r|r|r|r}
  n
  & \multicolumn1{c|}{15}
  & \multicolumn1{c|}{16}
  & \multicolumn1{c|}{17}
  & \multicolumn1{c|}{18}
  & \multicolumn1{c}{19} \\ \hline
  |G_n| & 10{,}349{,}536 & 46{,}206{,}736 & 211{,}799{,}312 & 997{,}313{,}824 & 4{,}809{,}701{,}440 \\
  |R(G_n)| & 3{,}780 & 2{,}627 & 10{,}187 & 6{,}422 & 26{,}796 \\ \hline
  \end{array}\]

  \caption{Values of $|G_n|$ and $|R(G_n)|$ for $n\leq 19$.
    Besides the values given in the table, $|R(G_{20})|=15{,}906$ was computed
    in a few seconds, and $|R(G_{30})|=1{,}248{,}696$ in under a minute.}
  \label{tab:Gn}
\end{table}

As mentioned in Section~\ref{sec:backgr}, $|G_n|$ corresponds to the number of
matchings in a complete graph with $n$ nodes, since every comparator joins two
channels.
The sequence $|R(G_n)|$ does not appear to be known already, and it does not
appear to have a simple description.
The following property is interesting to note.

\begin{theorem}\mbox{}\label{thm:interesting1}
  For odd $n$,
  $\left|R(G_n)\right|=\left|R(G_{n-1})\right|+2\left|R(G_{n-2})\right|$.
\end{theorem}
\begin{proof}
  The proof is based on the word representation of the networks.
   If $n$ is odd, then $\word(C)$ contains exactly one word beginning with
    \sent{0}.
    If this word is \sent{0_h}, then removing it yields a network with $n-1$
    channels, and this construction is reversible.
    Otherwise, removing the two last letters in this word yields a network with
    $n-2$ channels; since the removed letters can be \sent{12} or \sent{21},
    this matches each network on $n-2$ channels to two networks on $n$
    channels.\qedhere
\end{proof}

\subsection{Saturation}

We now introduce a notion that further restricts the set of networks
we need to consider when searching for optimal sorting networks.
Instead of just looking at the structure of the comparators (modulo
permutation), we further take into account the actual effect
of the network on its inputs, and focus only on
those networks that achieve the ``most'' amount of sorting. Similar to the
way that we use grammars to characterize isomorphic networks, here too
we first define the desired semantic property, and later provide a
syntactic characterization in terms of a grammar. 
The following lemma makes precise what we mean by achieving the ``most'' amount
of sorting.

\begin{lemma}
  \label{lemma:subset}
  Let $C=P;S$ be a sorting network of depth $d$ and $Q$ be a
  comparator network such that $P$ and $Q$ have the same depth, and
  $\outputs(Q)\subseteq\outputs(P)$.  Then $Q;S$ is a sorting network
  of depth $d$.
\end{lemma}
\begin{proof}
  Since $P$ and $Q$ have the same depth, the depth of $Q;S$ is $d$.
  Let $\vec x\in\BB^n$ be an arbitrary input.
  Then $Q(x)\in\outputs(Q)\subseteq\outputs(P)$.
  Hence, there is $y\in\BB^n$ such that $Q(x)=P(y)$.
  Thus, $(Q;S)(x) = S(Q(x)) = S(P(y)) = (P;S)(y) = C(y)$, which is sorted
  since $C$ is a sorting network.
\end{proof}

Lemma~\ref{lemma:subset} generalizes, so that it suffices that there
exists a permutation mapping the set of outputs of one network
into the set of outputs of the other network.

\begin{lemma}
  \label{lemma:symsubset}
  Let $C=P;S$ be a sorting network of depth $d$ and $Q$ be a
  comparator network such that $P$ and $Q$ have the same depth, and
  $\outputs(Q)\subseteq\pi(\outputs(P))$ for some permutation $\pi$ on
  $n$ channels.  Then there exists a sorting network of the form
  $Q;S'$ of depth $d$.
\end{lemma}

\begin{proof}
  Let $C=P;S$ be a sorting network of depth $d$.  Then
  $\pi(C)=\pi(P);\pi(S)$ is a generalized sorting network.  Since
  $\outputs(Q)\subseteq\pi(\outputs(P))=\outputs(\pi(P))$,
  Lemma~\ref{lemma:subset} implies that $Q;\pi(S)$ is also a
  generalized sorting network.
  Untangling $\pi(S)$, we obtain $S'$ such that $Q;S'$ is a
  sorting network of depth~$d$.  
\end{proof}

For comparator networks $C_a$ and $C_b$, if
$\outputs(C_b)\subseteq\pi(\outputs(C_a))$ for some permutation $\pi$,
we write $C_b\preceq C_a$, and say that $C_b$ \emph{subsumes} $C_a$.
Note that this relation includes equivalence.
By Lemma~\ref{lemma:symsubset}, it suffices to consider
two-layer networks that are minimal with respect to subsumption.

\begin{corollary}
  The set of equivalence classes of two-layer networks that are
  minimal with respect to subsumption is a complete set of filters.
\end{corollary}

Suppose one wishes to compute these minimal (up to the subsumes
relation) elements directly. Having fixed the first layer to some
maximal layer (e.g., $F_n$), there are still $|G_n|$
many possibilities for the second layer (see Table~\ref{tab:Gn}), the size
of their output sets is potentially exponential, and there are $n!$
permutations to consider, per pair of output sets, to determine
subsumption. So this problem quickly becomes intractable. One might
consider clever optimizations to reduce the computation time, but such
an approach does not scale well either.

Instead, we shall define a new class of networks, which we call
\emph{saturated two-layer networks}. This class, as it turns out, has
a simple syntactical characterization using the notion of
words. Moreover, we shall prove that it forms a complete set of filters. We
experimentally verify for small values of $n$ that the class of saturated 
two-layer networks is precisely the set of minimal two-layer networks with
respect to subsumption. We conjecture that the equality holds for all values
of $n$. Note that our results on the depth-optimality of sorting networks do
not depend on this conjecture as we only require that the class of saturated
two-layer prefixes forms a complete set of filters, which we prove in
Theorem~\ref{theorem:saturatedfilter}.

Before we introduce saturated networks formally, we point out that
restricting attention to such networks significantly reduces the
number of two-layer prefixes we need to consider.  The numbers
$|S_n|$ of saturated two-layer networks and $|R(S_n)|$ of their
equivalence classes modulo graph isomorphism are given in
Table~\ref{tab:Sn}. For example, for $n=16$, we reduce the number of
two-layer prefixes to consider from $2{,}627$ to
$323$.

\begin{definition}
  A comparator network $C$ is \emph{redundant} if there exists a
  network $C'$ obtained from $C$ by removing a comparator such that
  $\outputs(C')=\pi(\outputs(C))$ for some permutation $\pi$.
  A network $C$ is \emph{saturated} if it is non redundant, and every network
  $C'$ obtained by adding a comparator to $C$ satisfies
  $\outputs(C')\not\subseteq\pi(\outputs(C))$ for every permutation $\pi$.
\end{definition}

For example, any comparator network that contains comparators
between the same two channels at consecutive layers is redundant.
The notion of saturation is a generalization of Parberry's work in the first layer~\cite{Parberry89}:
Lemma~\ref{lem:parberry} can be restated as saying that the first layer
of a saturated comparator network on $n$ channels always contains
$\left\lfloor\frac n2\right\rfloor$ comparators.

The following property quantifies the impact of removing redundant two-layer
prefixes.

\begin{theorem}\mbox{}\label{thm:interesting2}
 The number of non-equivalent redundant two-layer networks on $n$
   channels is $|R(G_{n-2})|$.
\end{theorem}
\begin{proof}
  The proof is based on the word representation of the networks.  If $C$
  is a redundant net, then the sentence $\word(C)$ contains the word
  $\sent{12_c}$.  Removing one occurrence of this word yields a
  sentence corresponding to a network with $n-2$ channels.  This
  construction is reversible, so there are $|R(G_{n-2})|$ words
  corresponding to redundant networks on $n$ channels.
\end{proof}

In order to characterize saturated networks syntactically, we adopt
the notion of a pattern.  A \emph{pattern} $P$ is a partially
specified network: it is a set of channels connected by comparators,
but it may also include ``external'' comparators (represented as a
singleton node) that are connected to one channel in $P$ and one
channel not in $P$.
A comparator network $C$ contains a pattern $P$ of depth~$d$
on $m$ channels if there are a depth-$d$ prefix $C_1$ of $C$ and distinct
channels $c_1,\ldots,c_m$ of $C_1$ such that:
(i)~if $P$ contains a comparator between channels $i$ and $j$
at layer~$1\leq k\leq d$, then $C_1$ contains a comparator between
channels $c_i$ and $c_j$ at layer~$k$; (ii)~if $P$ contains an
external comparator touching channel $i$ at layer~$1\leq k\leq d$,
then $C_1$ contains a comparator between channel $c_i$ and a channel
$c\not\in\{c_1,\ldots,c_m\}$ at layer~$k$; (iii)~$C_1$ contains no
other comparators connecting to or between channels $c_1,\ldots,c_m$.

Figure~\ref{fig:patterns} depicts two patterns~$(a)$ and~$(b)$ and two
networks~$(c)$ and~($d)$.
The depth~$2$, $3$-channel pattern depicted in~$(a)$ occurs in network~$(c)$
but not in~$(d)$, while the pattern in~$(b)$ does not occur in
either network~$(c)$ or~$(d)$: its third channel is never used, while
all channels of~$(c)$ and~$(d)$ are used in the first two layers.

\begin{figure}[hb]
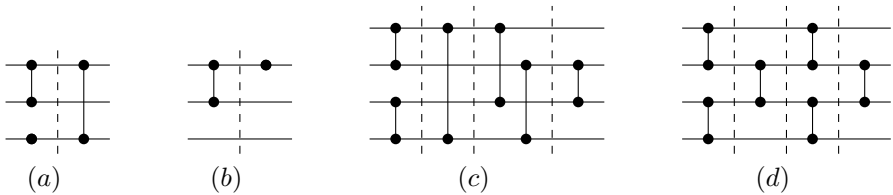

  \begin{tabular}{cccc}
    \begin{sortingnetwork}{3}{0.7}
      \addcomparator23
      \addcomparator11
      \nextlayer
      \addcomparator13
    \end{sortingnetwork}
    &
    \begin{sortingnetwork}{3}{0.7}
      \addcomparator23
      \nextlayer
      \addcomparator33
    \end{sortingnetwork}
    &
    \begin{sortingnetwork}{4}{0.7}
      \addcomparator12
      \addcomparator34
      \nextlayer
      \addcomparator14
      \nextlayer
      \addcomparator24
      \addlayer
      \addcomparator13
      \nextlayer
      \addcomparator23
    \end{sortingnetwork}
    &
    \begin{sortingnetwork}{4}{0.7}
      \addcomparator12
      \addcomparator34
      \nextlayer
      \addcomparator23
      \nextlayer
      \addcomparator12
      \addcomparator34
      \nextlayer
      \addcomparator23
    \end{sortingnetwork}
    \\
    $(a)$ & $(b)$ & $(c)$ & $(d)$
  \end{tabular}

  \caption{Two patterns~$(a,b)$ and two networks~$(c,d)$.}
  \label{fig:patterns}
\end{figure}
\begin{theorem}
  \label{thm:sat-char}
  Let $C$ be a saturated two-layer network.
  Then $C$ contains none of the two-layer patterns in Figure~\ref{fig:allpatterns}.
\end{theorem}

\begin{figure}[!hb]
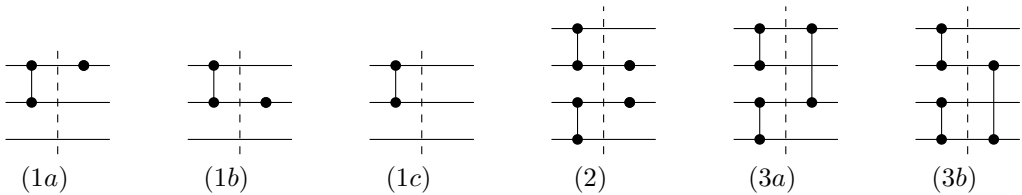

  \begin{tabular}{cccccc}
    \begin{sortingnetwork}{3}{0.7}
      \addcomparator23
      \nextlayer
      \addcomparator33
    \end{sortingnetwork}
    &
    \begin{sortingnetwork}{3}{0.7}
      \addcomparator23
      \nextlayer
      \addcomparator22
    \end{sortingnetwork}
    &
    \begin{sortingnetwork}{3}{0.7}
      \addcomparator23
      \nextlayer
    \end{sortingnetwork}
    &
    \begin{sortingnetwork}{4}{0.7}
      \addcomparator12
      \addcomparator34
      \nextlayer
      \addcomparator22
      \addcomparator33
    \end{sortingnetwork}
    &
    \begin{sortingnetwork}{4}{0.7}
      \addcomparator12
      \addcomparator34
      \nextlayer
      \addcomparator24
    \end{sortingnetwork}
    &
    \begin{sortingnetwork}{4}{0.7}
      \addcomparator12
      \addcomparator34
      \nextlayer
      \addcomparator13
    \end{sortingnetwork}
    \\
    $(1a)$ & $(1b)$ & $(1c)$ & $(2)$ & $(3a)$ & $(3b)$
  \end{tabular}
  \caption{Patterns forbidden in a saturated two-layer network.}
  \label{fig:allpatterns}
\end{figure}
\begin{proof}
  In all cases we show how to find $i$ and $j$
  such that $\outputs(C;(i,j))\subseteq\outputs(C)$ whenever $C$ contains
  one of the patterns in Figure~\ref{fig:allpatterns}.
  Depending on how the pattern is embedded in~$C$, $(i,j)$ may be a generalized
  comparator; in that case,
  $\outputs(C;(j,i))\subseteq\outputs(\pi(C))$ for $\pi=(i\,j)$.
  \begin{enumerate}
  \item Assume by contradiction that $C$ includes the pattern~$(1a)$
    and let the channels corresponding to those in the pattern be
    $a$, $b$ and~$c$, from top to bottom.
    Define $C'=C;(c,b)$ and let $\vec x\in\BB^n$.
    If $C'(\vec x)\neq C(\vec x)$, then $C(\vec x,1,c)=1$ and $C(\vec x,1,b)=0$.
    Since $b$ is a max-channel, this means that $C(\vec x,0,a)=C(\vec x,0,b)=0$
    and $C(\vec x,0,c)=1$.
    Then $C(\vec x')=C'(\vec x)$ for the input $\vec x'$ obtained from
    $\vec x$ by exchanging the values in positions~$c$ and~$b$.
    Therefore $\outputs(C')\subseteq\outputs(C)$, contradicting the fact
    that $C$ is saturated.

    Case~$(1b)$ is similar, adding a comparator $(a,c)$, and either
    construction applies to case~$(1c)$.

  \item Assume by contradiction that $C$ includes the pattern~$(2)$,
    and let the channels corresponding to those in the pattern
    be $a$, $b$, $c$ and~$d$, from top to bottom.
    Define $C'=C;(a,d)$, and let $\vec x\in\BB^n$.
    If $C'(\vec x)\neq C(\vec x)$, then $C(\vec x,1,a)=1$ and $C(\vec x,1,d)=0$.
    Since $a$ is a min-channel and $d$ is a max-channel, this means
    that $C(\vec x,0,a)=(\vec x,0,b)=1$ and $(\vec x,0,c)=(\vec x,0,d)=0$.
    Then $C(\vec x')=C'(\vec x)$ for the input $\vec x'$ obtained from
    $\vec x$ by exchanging the values in positions~$a$ and~$d$.
    Therefore $\outputs(C')\subseteq\outputs(C)$, contradicting the fact
    that $C$ is saturated.

  \item Assume by contradiction that $C$ includes the pattern~$(3a)$, and
    let the channels corresponding to those in the pattern be $a$, $b$, $c$
    and~$d$, from top to bottom.
    Define $C'=C;(b,d)$, and let $\vec x\in\BB^n$.
    If $C'(\vec x)\neq C(\vec x)$, then $C(\vec x,1,b)=1$ and $C(\vec x,1,d)=0$.
    Then $C(\vec x')=C'(\vec x)$ for the input $\vec x'$ obtained from
    $\vec x$ by exchanging the values in positions~$a$ and~$b$ with the
    values in positions~$c$ and~$d$, respectively.
    Note that this will permute $C(\vec x,1,a)$ and $C(\vec x,1,c)$,
    but it will not affect the final values on channels $a$ and $c$.
    Therefore $\outputs(C')\subseteq\outputs(C)$, contradicting the fact that
    $C$ is saturated.

    Case~$(3b)$ is similar.
  \end{enumerate}
  In all three cases, it is straightforward to verify that the inclusion
  $\outputs(C')\subseteq\outputs(C)$ is strict.
\end{proof}

In fact, the patterns in Figure~\ref{fig:allpatterns} are actually
\emph{all} of the patterns that make a comparator network non-saturated.
\begin{theorem}
  \label{thm:sat-thm}
  If $C$ is a non-redundant two-layer comparator network on
  $n$ channels containing none of the patterns in
  Figure~\ref{fig:allpatterns}, then $C$ is saturated.
\end{theorem}
\begin{proof}
  Let $C$ be a non-redundant two-layer comparator network,
  and assume that the second layer of $C$ has at least two unused channels
  (otherwise there is nothing to prove).
  If one of these channels were unused at layer~1, then the network would
  contain pattern~$(1a)$, $(1b)$ or~$(1c)$.
  Thus, the two channels are necessarily
  used in a comparator in layer~$1$ by Theorem~\ref{thm:sat-char}.
  From the same theorem, they must be both min-channels or both max-channels,
  otherwise the network would contain pattern~$(2)$;
  and the channels they are connected to at layer~$1$ cannot be connected at
  layer~$2$, otherwise the network would contain pattern~$(3a)$ or~$(3b)$.

  There are eight different cases to consider.
  We detail the cases where the two unused channels are max channels.
  Assume that the four relevant channels are adjacent, labeled $a$, $b$, $c$
  and $d$ from top to bottom, with first-layer comparators $(a,b)$ and $(c,d)$.
  This does not lose generality, but makes the presentation simpler: for the
  general case, just apply the permutation that brings any network to this
  particular form 
  to the reasoning below.
  This transformation can always be done preserving the
  standard comparator network form.

  Let $k$ be the number of channels above~$a$ and $m$ be the number of
  channels below~$d$.
  Let $C'$ be obtained from $C$ by adding the comparator $(b,d)$ at
  layer~$2$.
  The four possibilities depend on whether channels~$a$ and~$c$ are min- or
  max-channels at layer~$2$, and are represented in Figure~\ref{fig:sat-thm}.

  \begin{figure}[!ht]
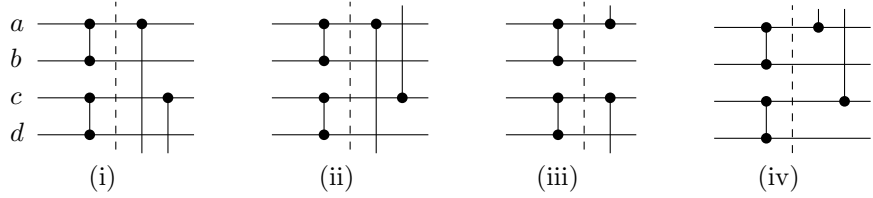

    \centering
    \begin{tabular}{cccc}
      \begin{sortingnetwork}{4}{0.7}
        \nodelabel{%
          \raisebox{-.6em}[0em][0em]{\makebox[0pt][l]{\hspace*{-2em}$d$}},%
          \raisebox{-.6em}[0em][0em]{\makebox[0pt][l]{\hspace*{-2em}$c$}},%
          \raisebox{-.6em}[0em][0em]{\makebox[0pt][l]{\hspace*{-2em}$b$}},%
          \raisebox{-.6em}[0em][0em]{\makebox[0pt][l]{\hspace*{-2em}$a$}}}
        \addcomparator{3}{4}
        \addcomparator{1}{2}
        \nextlayer
        \addhalfcomparator{4}{0.5}
        \addlayer
        \addhalfcomparator{2}{0.5}
      \end{sortingnetwork}
      &
      \begin{sortingnetwork}{4}{0.7}
        \addlayer
        \addcomparator{3}{4}
        \addcomparator{1}{2}
        \nextlayer
        \addhalfcomparator{4}{0.5}
        \addlayer
        \addhalfcomparator{2}{4.5}
      \end{sortingnetwork}
      &
      \begin{sortingnetwork}{4}{0.7}
        \addlayer
        \addcomparator{3}{4}
        \addcomparator{1}{2}
        \nextlayer
        \addhalfcomparator{4}{4.5}
        \addhalfcomparator{2}{0.5}
      \end{sortingnetwork}
      &
      \begin{sortingnetwork}{4}{0.7}
        \addlayer
        \addcomparator{3}{4}
        \addcomparator{1}{2}
        \nextlayer
        \addhalfcomparator{4}{4.5}
        \addlayer
        \addhalfcomparator{2}{4.5}
      \end{sortingnetwork}
      \\
      (i) & (ii) & (iii) & (iv)
    \end{tabular}

    \caption{Possible cases for channels~$a$ and~$c$ in the proof of
      Theorem~\ref{thm:sat-thm}.
      To obtain $C'$, add a comparator between channels~$b$ and $d$.}
    \label{fig:sat-thm}
  \end{figure}

  \begin{description}
  \item[(Min/min)] Channel~$a$ and $b$ are min-channels at layer~$2$, so the
    network looks as in Case~(i) of Figure~\ref{fig:sat-thm}.

    Consider the input string $1^k11001^m$.
    Since $C'(1^k11001^m)=1^k10011^m$, we have $1^k10011^m\in\outputs(C')$.
    We now show that $1^k10011^m\not\in\outputs(C)$.
    Because of the comparator~$(a,b)$ at layer~$1$, to obtain the $0$ on
    channel~$b$ the input string would necessarily have a~$0$ on channel~$a$.
    But then the output would also have a~$0$ on channel~$a$, hence it could
    not be $1^k10011^m$.

  \item[(Min/max)] The network now looks as in Case~(ii) of
    Figure~\ref{fig:sat-thm}, and the argument is similar.
    Choose an input $\vec x$ such that $C'(\vec x)$ has $1001$ on channels
    $a$--$d$; this is possible by placing $0$s on the two positions that may
    be compared to~$c$ at layer~$2$ and $1$s on the two positions that may be
    compared to~$a$ at layer~$2$.
    Then the argument from the previous case again shows that
    $C'(\vec x)\not\in\outputs(C)$.

  \item[(Max/min)] The network now looks as in Case~(iii) of
    Figure~\ref{fig:sat-thm}.  This can be reduced to the previous case
    by interchanging $a$ and $b$ with $c$ and $d$, respectively.

  \item[(Max/max)] The network now looks as in Case~(iv) of
    Figure~\ref{fig:sat-thm} and the reasoning is a bit more involved.

    Consider once more the input string $1^k11001^m$.
    Channel~$c$ is now a second-layer max-channel connected to some
    channel~$j\leq k$, and~$C'(1^k11001^m)=1^{j-1}01^{k-j}10111^m$.
    In order to obtain this output with network~$C$, it is again necessary to
    have inputs $0$ on channels~$a$ and~$b$; but since there are only two $0$s
    in the output, this means that channel~$a$ must also be connected to
    channel $j$ on layer~$2$, which is impossible.
  \end{description}

  The cases where $a$ and $c$ are the unused (min) channels are similar.
  \begin{description}
  \item[(Min/min)] Similar to the case (Max/max) above, using the input string
    $0^k11000^m$ and analyzing the result on channel~$c$.

  \item[(Min/max)] Similar to the case (Min/max) above, using an input string
    that produces an output of the form $v1001w$, and analyzing the result on
    channel~$c$.

  \item[(Max/min)] This can be reduced to the previous case
    by interchanging $a$ and $b$ with $c$ and $d$, respectively.

  \item[(Max/max)] Similar to the case (Min/min) above, using the input string
    $0^k11000^m$, and analyzing the result on channel $c$.\qedhere
  \end{description}
\end{proof}

As a corollary of Theorems~\ref{thm:sat-char} and~\ref{thm:sat-thm}, we
show that we can always assume the first two layers of a sorting
network to be saturated.

\begin{corollary}\label{corollary:saturated}
  Let $L_1$ and $L_2$ be layers on $n$ channels such that $L_1$ is
  maximal.  Then there is a layer $S$ such that
  $\outputs(L_1;S)\subseteq \outputs(L_1;L_2)$ and $L_1;S$ is
  saturated.
\end{corollary}
\begin{proof}
  By removing comparators if necessary, we can assume that $L_1;L_2$
  is nonredundant.
  If $L_1;L_2$ is not saturated, then, by Theorem~\ref{thm:sat-thm},
  it must contain some of the patterns from
  Figure~\ref{fig:allpatterns}. Now, for each pattern occurring in $L_1;L_2$, the
  argument in the proof of Theorem~\ref{thm:sat-char} tells us how to
  eliminate it by adding comparators to $L_2$. Denote the
  obtained layer by $S$. The argument in the proof of
  Theorem~\ref{thm:sat-char} further ensures that

  By construction, the network $L_1;S$ does not contain any of the
  patterns from Figure~\ref{fig:allpatterns}, and so, by
  Theorem~\ref{thm:sat-thm}, it is saturated.
\end{proof}

The above Corollary together with Lemma~\ref{lemma:symsubset} imply
that if there is a sorting network of a given size, then there is one
whose first two layers are saturated, i.e., the set of all saturated two-layer networks
is a complete set of filters.

\begin{theorem}
\label{theorem:saturatedfilter}
For every $n$, both the set $S_n$ of two-layer saturated networks and the
set of representatives of its
equivalence classes $R(S_n)$ are complete sets of filters on $n$
channels.
\end{theorem}
\begin{proof}
  Corollary~\ref{corollary:saturated} and Lemma~\ref{lemma:symsubset}
  imply that $S_n$ is a complete set of
  filters. By Lemma~\ref{lem:isomorphism}, $R(S_n)$ is
  also a complete set of filters.
\end{proof}

We conjecture that in fact the following result holds.
\begin{conjecture}\label{conjecture}
  If networks $C_1$ and $C_2$ on $n$ channels are both saturated and
  non-equivalent, then $\outputs(C_1)\not\subseteq(C_2)$ for any permutation $\pi$.
\end{conjecture}

Particular cases of Conjecture~\ref{conjecture} are implied by
Theorem~\ref{thm:sat-thm}, but the general case remains open. The
conjecture has been verified experimentally for $n\leq 15$.

The characterization of saturation given by Theorem~\ref{thm:sat-thm} is
straightforward to translate in terms of the word associated with a network.
\begin{corollary}
  Let $C$ be a two-layer network.
  Then $C$ is saturated if $w=\word(C)$ satisfies the following properties.
  \begin{enumerate}
  \item If $w$ contains \sent{0_h} or $\sent{12_s}$, then all other words
    in $w$ are cycles.
  \item No stick $w$ has length $4$.
  \item Every stick in $w$ begins and ends with the same symbol.
  \item If $w$ contains a head or stick ending with $x$, then every head or stick in $w$ ends with $x$, for $x\in\{\sent{1},\sent{2}\}$.
  \end{enumerate}
\end{corollary}

\noindent Thus, the set of saturated two-layer networks can be generated by using the
following restricted grammar.
\begin{align}
  \label{eq:satword}
  \mathsf{Word} &::= \mathsf{Head} \mid \mathsf{Stick} \mid \mathsf{Cycle}
  & \mathsf{Stick} &::= \sent{12} \mid \mathsf{eStick} \mid \mathsf{oStick} \\
  \nonumber\mathsf{Head} &::= \sent0 \mid \mathsf{eHead} \mid \mathsf{oHead}
  & \mathsf{eStick} &::= \sent{12}(\sent{12}+\sent{21})^+\sent{21} \\
  \nonumber\mathsf{eHead} &::= \sent0(\sent{12}+\sent{21})^\ast\sent{12}
  & \mathsf{oStick} &::= \sent{21}(\sent{12}+\sent{21})^+\sent{12} \\
  \nonumber\mathsf{oHead} &::= \sent0(\sent{12}+\sent{21})^\ast\sent{21}
  & \mathsf{Cycle} &::= \sent{12}(\sent{12}+\sent{21})^+
\end{align}

\noindent Furthermore, sentences are multi-sets $M$ such that:
\begin{itemize}
\item if $M$ contains the words \sent{0_h} or \sent{12_s},
  then all other elements of $M$ are cycles;
\item if $M$ contains an $\mathsf{eHead}$ or $\mathsf{eStick}$, then it
  contains no $\mathsf{oHead}$ or $\mathsf{oStick}$.
\end{itemize}
With these restrictions, generating all saturated networks for $n\leq
20$ can be done almost instantaneously.  The numbers $|S_n|$ of
saturated two-layer networks and $|R(S_n)|$ of equivalence classes
modulo permutation are given in rows two and four of
Table~\ref{tab:Sn}.

\begin{table}
  \[\begin{array}{c|r|r|r|r|r|r|r|r|r|r|r|r}
  n & \multicolumn1{c|}{3}
  & \multicolumn1{c|}{4}
  & \multicolumn1{c|}{5}
  & \multicolumn1{c|}{6}
  & \multicolumn1{c|}{7}
  & \multicolumn1{c|}{8}
  & \multicolumn1{c|}{9}
  & \multicolumn1{c|}{10}
  & \multicolumn1{c|}{11}
  & \multicolumn1{c|}{12}
  & \multicolumn1{c|}{13}
  & \multicolumn1{c}{14}  \\ \hline
  |G_n| & 4 & 10 & 26 & 76 & 232 & 764 & 2{,}620 & 9{,}496 & 35{,}696 & 140{,}152 & 568{,}504 & 2{,}390{,}480 \\
  |S_n| & 2 & 4 & 10 & 28 & 70 & 230 & 676 & 2{,}456 & 7{,}916 & 31{,}374 & 109{,}856 & 467{,}716 \\
  |R(G_n)| & 4 & 8 & 16 & 20 & 52 & 61 & 165 & 152 & 482 & 414 & 1{,}378 & 1{,}024 \\
  |R(S_n)| & 2 & 2 & 6 & 6 & 14 & 15 & 37 & 27 & 88 & 70 & 212 & 136 \\ 
  |R_n| & 1 & 2 & 4 & 5 & 8 & 12 & 22 & 21 & 48 & 50 & 117 & 94 \\ 
  \hline
  \end{array}\]
  \[\begin{array}{c|r|r|r|r|r}
  n
  & \multicolumn1{c|}{15}
  & \multicolumn1{c|}{16}
  & \multicolumn1{c|}{17}
  & \multicolumn1{c|}{18}
  & \multicolumn1{c}{19} \\ \hline
  |G_n| & 10{,}349{,}536 & 46{,}206{,}736 & 211{,}799{,}312 & 997{,}313{,}824 & 4{,}809{,}701{,}440 \\
  |S_n| & 1{,}759{,}422 & 7{,}968{,}204 & 31{,}922{,}840 & 152{,}664{,}200 & 646{,}888{,}154\\
  |R(G_n)| & 3{,}780 & 2{,}627 & 10{,}187 & 6{,}422 & 26{,}796 \\
  |R(S_n)| & 494 & 323 & 1{,}149 & 651 & 2{,}632 \\ 
  |R_n| & 262 & 211 & 609 & 411 & 1{,}367 \\ 
  \hline
  \end{array}\]

  \caption{Values of $|G_n|$, $|R(G_n)|$, $|S_n|$, $|R(S_n)|$ and $|R_n|$ for
    $n\leq 19$.}
  \label{tab:Sn}
\end{table}

\subsection{Reflections}

In the previous section, we have shown that it suffices to consider
only representatives of saturated networks, where two (saturated)
networks are equivalent if their corresponding graphs are
isomorphic. In this section, we extend the notion of equivalence by
noting that a (vertical) reflection of a sorting network is also a
sorting network.

Formally, the reflection of a comparator network $C$ on $n$ channels is the
network $C^R$ obtained from $C$ by replacing each comparator $(i,j)$ with
$(n-j+1,n-i+1)$. Note that this operation preserves the size and the depth of $C$.
Figure~\ref{fig:reflection} shows a two-layer, $6$-channel comparator network
and its reflection.
These networks are both elements of $R(S_6)$, corresponding to the two
different words~\sent{211212_s} and~\sent{121221_s}, and as such are not
equivalent using the theory developed so far.

\begin{figure}
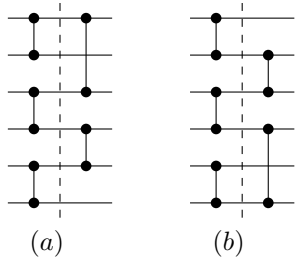

  \centering
  \begin{tabular}{cc}
  \begin{sortingnetwork}{6}{0.7}
    \addcomparator12
    \addcomparator34
    \addcomparator56
    \nextlayer
    \addcomparator23
    \addcomparator46
  \end{sortingnetwork}
  &
  \begin{sortingnetwork}{6}{0.7}
    \addcomparator12
    \addcomparator34
    \addcomparator56
    \nextlayer
    \addcomparator13
    \addcomparator45
  \end{sortingnetwork}
  \\
  $(a)$ & $(b)$
  \end{tabular}

  \caption{A two-layer, $6$-channel comparator network~$(a)$ and its reflection~$(b)$.}
  \label{fig:reflection}
\end{figure}

Given a vector $\vec x\in\BB^n$, denote by $\overline{\vec x^R}$ the
vector obtained from $\vec x$ by reversing and complementing each bit.
For example, $\overline{100^R} = 110$.

\begin{lemma}
  \label{lem:reflection}
  Let $C$ be a comparator network on $n$ channels and $C^R$ be its reflection.
  Then $\vec x\in\outputs(C)$ if and only if~~$\overline{\vec x^R}\in\outputs(C^R)$.
\end{lemma}
\begin{proof}
  By induction on the size of $C$.
  For the empty network the result is trivial.
  Assume the result holds for $C$ and consider the network $C;(i,j)$.
  Let $\vec x\in\BB^n$ and $\vec y=C(\vec x)$; the induction hypothesis guarantees that
  $\overline{\vec y^R}\in\outputs(C^R)$.
  Then, the comparator $(i,j)$ will change $\vec y$ if and only if the comparator $(n-j+1,n-i+1)$ changes $\overline{\vec y^R}$,
  interchanging the corresponding positions in both sequences.  This establishes the thesis for $C;(i,j)$.
\end{proof}

It follows from Lemma~\ref{lem:reflection} that $C$ can be extended to a sorting network of some depth $d$ if and only if $C^R$ can be extended to a sorting network of depth $d$.
Thus, we can further reduce the number of candidate two-layer prefixes
by eliminating those that are reflections of others.

\begin{corollary}
Let $S_n'$ be any subset of saturated two-layer networks on $n$ channels containing only one of $C$ and $C^R$ for each $C \in S_n$. Then both $S_n'$ and $R(S_n')$ form complete set of filters.
\end{corollary}

Yet again, particular sets $S_n'$ and $R(S_n')$ can be constructed syntactically by considering the word representation.
Let $C$ be a saturated network such that $\word(C)$ contains at least one Head
or Stick word.
Reflection transforms min-channels into max-channels and vice-versa, so the
reflection of an \textsf{oHead} (respectively \textsf{oStick}) is an
\textsf{eHead} (resp.\ \textsf{eStick}), and conversely.
We can thus restrict the grammar defining saturated two-layer networks~\eqref{eq:satword}
to the following,
which allows neither \textsf{eHead}s nor \textsf{eStick}s.
\begin{align}
  \label{eq:reflword}
  \mathsf{Word} &::= \sent{0} \mid \mathsf{oHead} \mid \sent{12} \mid \mathsf{oStick} \mid \mathsf{Cycle} \\
  \nonumber\mathsf{oStick} &::= \sent{21}(\sent{12}+\sent{21})^+\sent{12} \\
  \nonumber\mathsf{oHead} &::= \sent0(\sent{12}+\sent{21})^\ast\sent{21} \\
  \nonumber\mathsf{Cycle} &::= \sent{12}(\sent{12}+\sent{21})^+
\end{align}

This handles the case of Head- and Stick-words. It remains to consider the reflections of Cycle-words.
Since reflection transforms min-channels into max-channels and conversely,
the word corresponding to the reflection of a cycle can be obtained by
interchanging \sent{1}s and \sent{2}s in the word corresponding to the cycle,
and then shifting and possibly reversing the result to obtain the
lexicographically smallest representative of that
cycle.
As it turns out, this will typically yield the original word;
in particular, for $n<12$, this is always the case, as the lemma below states.

Given a word $w$, denote by $\overline w$ the word obtained by interchanging
\sent{1}s and \sent{2}s in $w$, and by $w^R$ the reverse word to $w$.

\begin{lemma}
  Let $C$ be a network on an even number $n$ of channels consisting of a
  connected cycle.
  If $n < 12$, then $C$ is equivalent to $C^R$.
\end{lemma}
\begin{proof}
  Every Cycle-word can be written as
  $w=\sent{12(12)}^{k_1}\sent{21(12)}^{k_2}\sent{21}\ldots\sent{21(12)}^{k_n}$.
  Then $\overline w^R$ is the word
  $\sent{(12)}^{k_n}\sent{21}\ldots\sent{21(12)}^{k_2}\sent{21(12)}^{k_1}\sent{12}$,
  and it can always be shifted into $w$ unless $k_1$, $k_2$ and $k_3$ are all
  distinct.
  The shortest word where $k_1$, $k_2$ and $k_3$ are all distinct
  is \sent{122112211212}, where $k_1=0$, $k_2=1$ and
  $k_3=2$, corresponding to a cycle on~$12$ channels.
\end{proof}

Call a word $w$ \emph{asymmetric} if $\word(\net(w)^R)\neq w$.
The previous result states that the shortest asymmetric Cycle-word has
length~$11$.
Table~\ref{tab:asymcycle} indicates the number $A_n$ of asymmetric
cycles on $n$~channels, modulo reflection.
These values are computed by generating all Cycle-words of length~$n$ using
the grammar in Equation~\eqref{eq:reflword} and testing
whether they are asymmetric.

This sequence has been described previously in~\cite{McLarnan81},
which describes the computation of the number of possible crystal
structures with particular kinds of symmetries. Sequence $A_n$ above
corresponds precisely to the cardinality of the space group $P6_3/mc$,
which is computed by a symbolic representation whose specification
exactly matches that of $A_n$.  The values given in~\cite{McLarnan81}
however differ from ours for the values $36+6k$, for integers $k\geq
0$.  We believe that this is due to errors in the original computations
described in~\cite{McLarnan81}.

\begin{table}[hb]
  \[\begin{array}{c|c|c|c|c|c|c|c|c|c|c|c|c|c|c|c}
    n & 12 & 14 & 16 & 18 & 20 & 22 & 24 & 26 & 28 & 30 & 32 & 34 & 36 & 38 & 40 \\ \hline
    A_n & 1 & 1 & 4 & 7 & 18 & 31 & 70 & 126 & 261 & 484 & 960 & 1{,}800 & 3{,}515 & 6{,}643 & 12{,}852
    \end{array}\]
  \caption{Number of asymmetric Cycle-words on $n$~channels, modulo reflection.}
\label{tab:asymcycle}
\end{table}

\begin{definition}
  \label{defn:Rn}
  The set of two-layer representative prefixes $R_n$ is the set of all
  networks generated from sentences $s$
  such that:
  \begin{itemize}
  \item all words $w\in s$ are generated from the grammar in Equation~\eqref{eq:reflword};
  \item all Stick-words $w\in s$ satisfy $w<w^R$;
  \item all Cycle-words $w\in s$ satisfy $w=\word(\net(w))$;
  \item if $s$ contains \sent{0_h} or \sent{12_s},
    then all other words in $s$ are Cycle-words;
  \item if $s$ does not contain \sent{oHead} or \sent{oStick} words
    and $k$ is the shortest length for
    which $s$ contains only one asymmetric Cycle-word $w$ of length~$k$
    (possibly with high multiplicity), then $w<\word(\net(w)^R)$.
  \end{itemize}
\end{definition}

\begin{theorem}
  \label{thm:refl-complete}
  The set $R_n$ contains a representative for all equivalence classes of
  networks on $n$ channels up to reflection.
  Furthermore, if $R_n$ contains two networks that are equivalent up to
  reflection, then those networks are represented by a sentence $s$ such that:
  \begin{itemize}
  \item $s$ contains only Cycle-words;
  \item for every length $k$, the number of distinct asymmetric words of
    length $k$ in $s$ is not~$1$.
  \end{itemize}
\end{theorem}
\begin{proof}
  By construction, $R_n$ contains one network from each equivalence class in
  $G_n$ modulo reflection.
  Furthermore, since \textsf{oStick}s and \textsf{oHead}s are reflected
  to \textsf{eStick}s and \textsf{eHead}s, no network containing any of these
  words can be represented together with its reflection.
  Likewise, if a network containing only cycles contains only one
  asymmetric cycle of some length, then the last criterium will guarantee that
  either itself or its reflection will not be included in $R_n$.
\end{proof}

\begin{corollary}\label{lem:complete}
  For every $n$, the set $R_n$ is a complete
  sets of filters on $n$ channels.
\end{corollary}

It is possible to find two asymmetric cycles $A$ and $B$ of the same length
such that $\word(A)<\word(B)$ and $\word(B^R)<\word(A^R)$, hence for some~$n$
the set $R_n$ still may contain redundancy.
However, any network with two such cycles has at least $32$ channels, since
each such cycle requires at least $16$ channels.
In particular, for $n\leq 31$ the sets $R_n$ contain no redundancy.

The last line of Table~\ref{tab:Sn} shows the number of
representatives, $|R_n|$,
one needs to analyze to solve the optimal-depth problem for sorting networks
of size up to $19$.
We can compute the sets $R_n$ efficiently (in under two hours) for $n\leq 40$.
For greater $n$, the execution time begins to grow due to the complexity of
the test for asymmetric cycles.
This could be reduced by techniques such as tabling; however,
the interest of such optimizations is limited, since $40$
is well beyond the scope of current techniques for solving the optimal-depth
problem.

The sets of comparator networks, $R_n$, for $n\leq 16$ can be downloaded from\\
\url{http://www.cs.bgu.ac.il/~mcodish/Papers/Appendices/SortingNetworks/}.

\section{Propositional Encoding of Sorting Networks}
\label{sec:encode}

In the previous section, we showed how to compute a set of two-layer
networks, $R_n$, that is complete when looking for optimal-depth
sorting networks on $n$ channels. In this section, we employ this
result to represent existence of sorting networks of
a given depth by propositional formulas. 
Using a SAT solver on the
obtained formulas, we can both find the optimal-depth
sorting networks for $n \leq 16$
channels and prove their optimality.

For experimentation we used a cluster of Intel E8400 cores clocked at
$2$~GHz each. Each of the cores in the cluster has computational power
comparable to a core on a standard desktop computer.  Although we used all of the cores of the cluster in our experimentation,
each individual instance was run on a single core.
All of the times
indicated in all of our results, detailed in the following, are
obtained using a single core of the cluster.

Morgenstern et al.~\cite{MorgensternS11} observed that an $n$-channel
comparator network of a given depth~$d$ can be represented by a
propositional formula such that the existence of a depth-$d$,
$n$-channel, sorting network is equivalent to that formula's satisfiability. We
improve upon the work in~\cite{MorgensternS11}, and give a more natural
translation to propositional formulas. In contrast to the encoding
proposed in~\cite{MorgensternS11}, which did not prove sufficient to
solve any open instances beyond $n = 10$, ours facilitates the search
for optimal-depth sorting networks with up to $16$ channels.

The encoding uses (similarly to the one in~\cite{MorgensternS11}) the
zero-one principle (Lemma~\ref{lemma:zeroone}), which states that,
when checking whether a comparator network is a sorting network,
it suffices to consider only its outputs on
Boolean inputs.
We can represent the effect of a comparator on Boolean values $x, y \in \BB$
as $\min(x, y) = x \wedge y$ and $\max(x, y) = x \vee y$.

\subsection{Optimal-depth sorting networks}

We now describe the construction of a propositional formula
$\Psi(n,d)$ that is satisfiable if and only if an $n$-channel sorting
network of depth~$d$ exists. Moreover, the formula $\Psi(n,d)$ 
has the property that: if the formula is satisfiable, then an
$n$-channel sorting network of depth~$d$ can be easily extracted from
a satisfying assignment.

The formula uses the following set of Boolean variables, specifying
the position of the comparators in the network:
\[
  \networkVars^d_n=
       \sset{c^\ell_{i,j}}{1\leq\ell\leq d,1\leq i<j\leq n }
\] 
where the intention is that $c^\ell_{i,j}$ is true if and only if the
network contains a comparator between channels $i$ and $j$ at depth
$\ell$.

Further, to facilitate the specification of the encoding, we introduce
an additional set of Boolean variables capturing which channels are
``used'' at a given layer:
\[
  \usedVars^d_n=
       \sset{u^\ell_k}{1\leq\ell\leq d,1\leq k\leq n }
\]
where the intention is that $u^\ell_k$ is true if and only if there is
some comparator on channel~$k$ at level $\ell$. 

The following formula enforces the relation between the
variables in $\usedVars^d_n$ and in $\networkVars^d_n$:
\[
    \varphi^{\emph{used}}_{n,d} ~~= 
          \bigwedge_{\scriptsize\begin{array}{l}
             1 \leq \ell \leq d\\
             1 \leq k \leq n
           \end{array}
          }
                \textstyle u^\ell_k \leftrightarrow 
                \bigvee incident^\ell_k(\networkVars^d_n)
\]
where  $incident^\ell_k(\networkVars^d_n)$ denotes, for each
channel $k$ and level $\ell$, which variables in $\networkVars^d_n$
correspond to comparators incident to channel $k$ at layer $\ell$:
\[
  incident^\ell_k(\networkVars^d_n)=\sset{c^\ell_{i,j}}{c^\ell_{i,j}\in
  \networkVars^d_n,~ i=k \mbox{~or~} j=k}\]

Now, a representation of a comparator network is valid if at each
layer every channel is used by at most one comparator. This is
enforced by the following formula:
\begin{eqnarray*}
\varphi^{valid}_{n,d} & = & \hspace{-5mm}
\bigwedge_{\scriptsize\begin{array}{l}
             1 \leq \ell \leq d\\
             1 \leq k \leq n
           \end{array}}\hspace{-3mm}    \once(incident^\ell_k(\networkVars^d_n))
\end{eqnarray*}
where for any set of Boolean variables $B=\{b_1,\ldots,b_n\}$, the
formula $\once(B)$ signifies that at most one of the variables
in $B$ takes the value $true$. We adopt the straightforward encoding:
\begin{eqnarray*}
\once(\{b_1,\ldots,b_n\}) & = & 
    \bigwedge_{i < j} (\neg b_i \lor \neg b_j) \\
\end{eqnarray*}

Given sets of Boolean variables $\bar x = \{x_1, \ldots, x_n\}$ and
$\bar y = \{y_1, \ldots, y_n\}$, we use the following formula to
express that $\{y_1, \ldots, y_n\}$ is obtained from $\{x_1, \ldots,
x_n\}$ by applying the $\ell$-th layer of the network:
\[
\varphi^\ell_{n,d}(\bar x,\bar y) = 
            \underbrace{\bigwedge_{i<j} c^\ell_{i,j} \rightarrow
            \left(\bigwedge \begin{array}{l} 
                  y_i\leftrightarrow x_i\land x_j\\
                  y_j\leftrightarrow x_i\lor x_j
                \end{array}\right)}_a
           ~~\wedge~~
           \underbrace{\bigwedge_{k} \neg u^\ell_k \rightarrow 
                    (x_k \leftrightarrow y_k)}_b
\]
The left part ($a$) specifies that the outputs of a
comparator on channels $i$ and $j$ are the minimum and maximum of
its inputs; the right part ($b$) specifies that, if a channel is not
incident to any comparator, then its output is equal to its input.

To express that the network sorts the input $\bar
b=\{b_1,\ldots,b_n\}$, we introduce Boolean variables $\bar x^i =
\{x^i_1,\ldots,x^i_n\}$ for $0\leq i\leq d$, where $\bar x^i$ shall
denote the values on the $n$ channels after the $i$th layer of the
network. We set $\bar x^0 = \bar b$, denote by $\bar b'$ the result of
sorting the given vector $\bar b$, and write:
\[ \varphi^{sort}_{n,d}(\bar b) = 
       \bigwedge_{\ell=1}^{d} 
       \varphi^\ell_{n,d}(\bar x^{i-1},\bar x^{i}) \wedge \bigwedge_{i=1}^n
       \bar x^d_i \leftrightarrow \bar b'_i
\]
Then, given a set of Boolean inputs $X \subseteq \BB^n$, the following
formula is satisfiable if and only if there is a depth-$d$,
$n$-channel network sorting all inputs from $X$:
\begin{equation}
\label{satEncoding1}
  \Psi(n,d,X)
     = \varphi^{\emph{used}}_{n,d}  \land  \varphi^{valid}_{n,d} \land
     \bigwedge_{{\bar b}\in X} \varphi^{sort}_{n,d}(\bar b)
\end{equation}
A sorting network must sort all Boolean inputs.
Hence, the following result holds.

\begin{lemma}
\label{lem:sat}
There exists a sorting network with $n$ channels and depth $d$ if and only
if the formula $\Psi(n,d, \BB^n)$ is satisfiable.
\end{lemma}

Later we show how to modify an encoding based on the formula in
Equation~(\ref{satEncoding1}) to make use of the results from
Section~\ref{sec:prefix}.
First, to improve the performance of an actual SAT solver on 
such an encoding, we introduce several additional optimizations which
we now briefly describe.

\begin{description}
\item[No redundant comparators.]
If a comparator occurs in two consecutive layers of a network, then
the second one has no effect. To prevent the placement of redundant
comparators, we introduce the following symmetry breaking formula:
\[
\sigma_1 = \hspace{-5mm}
           \bigwedge_{\scriptsize\begin{array}{c}
             1 \leq \ell < d\\
             1 \leq i<j \leq n
           \end{array}}\hspace{-3mm}
           \neg c^\ell_{i,j} \vee \neg c^{\ell+1}_{i,j}
\]
\item[Eager comparator placement.]
If a comparator is positioned on a pair of channels at level $\ell$
that are not used at level $\ell-1$, then it can be ``slided to the previous layer''.
To prevent
the placement of such sliding comparators, we introduce the following
symmetry breaking formula:
\[
\sigma_2 = \hspace{-5mm}
           \bigwedge_{\scriptsize\begin{array}{c}
             1 < \ell \leq d\\
             1 \leq i<j \leq n
           \end{array}}\hspace{-3mm}
           c^\ell_{i,j} \rightarrow u^{\ell-1}_{i} \vee u^{\ell-1}_{j}
\]
\item[All adjacent comparators.]
Exercise 5.3.4.35 in~\cite{Knuth73} states that all comparators of the
form $(i,i+1)$ must be present in a sorting network. To this end, we
add the following (redundant) formula:
\[
\sigma_3 = \hspace{-5mm}
           \bigwedge_{\scriptsize\begin{array}{c}
             1 \leq i < n
           \end{array}}\hspace{-3mm}
           \left( c^1_{i,i+1} \vee c^2_{i,i+1} \vee\cdots\vee c^d_{i,i+1}
           \right)
\]
\item[Only unsorted inputs.]
Let $\BB^n_{un}$ denote the subset of $\BB^n$ consisting of unsorted
sequences. Then it is possible to refine the formula in
Lemma~\ref{lem:sat} by replacing the set of all Boolean inputs $\BB^n$
with the set of unsorted inputs~$\BB^n_{un}$. This is the case as
sorted sequences are unchanged regardless of the positioning of the
comparators.  Observe that $|\BB^n_{un}| = 2^n-n-1$, and as noted by
Chung and Ravikumar~\cite{Chung1990}, this is the size of the smallest
test set possible needed to determine whether a comparator network is
a sorting network.
\item[Optimized CNF generation.]
Our encodings are generated using the \bee\ finite-domain constraint
compiler. \bee\ is described in several recent
papers~\cite{bee2012,Metodi2011,jair2013}. \bee\
facilitates solving finite-domain constraints by encoding them to CNF
and applying an underlying SAT solver. In \bee\ constraints are
modeled as Boolean functions which propagate information about
equalities between Boolean literals.  This information is then applied
to simplify the CNF encoding of the constraints. \bee\ is written in
Prolog, and applies (in our configuration) the underlying SAT solver
CryptoMiniSAT~\cite{Crypto}.
\end{description}

\subsubsection*{Experiment I}

In our first experiment, we used the encoding of
Equation~\eqref{satEncoding1}. Here, the formula $\Psi(n,d,X)$,
together with the above described optimizations, is instantiated for
various values of $n, d$ and $X$.
Table~\ref{fig:satPlain} presents the results, where each
instance (2 per line in the table) is run on a single thread of the
cluster. 
\begin{table}[b]
\centering\scriptsize\begin{tabular}{|l||r|r|r|r|r||r|r|r|r|r||}
\cline{2-11} \multicolumn{1}{c||}{}
    &\multicolumn{5}{c||}{optimal sorting networks (sat)}      
    &\multicolumn{5}{c||}{smaller networks (unsat)}\\
\hline
$n$ &$d$&
\multicolumn1{c|}{\bee} &
\multicolumn1{c|}{\#clauses} &
\multicolumn1{c|}{\#vars} &
\multicolumn1{c||}{SAT} &
$d'$&
\multicolumn1{c|}{\bee} &
\multicolumn1{c|}{\#clauses} &
\multicolumn1{c|}{\#vars} &
\multicolumn1{c||}{SAT} \\
\hline
 5 & 5 &    0.09 &    4965 &   761 &    0.01 & 4 &   0.08 &     3702 &    550 &    0.01 \\
 6 & 5 &    0.31 &   15353 &  1911 &    0.04 & 4 &   0.23 &    11417 &   1374 &    0.03 \\
 7 & 6 &    1.14 &   55758 &  5946 &    0.14 & 5 &   0.83 &    44330 &   4634 &    0.97 \\
 8 & 6 &    3.55 &  153125 & 14058 &    1.35 & 5 &   2.47 &   121639 &  10946 &    1.83 \\
 9 & 7 &   10.06 &  487489 & 39761 &    9.51 & 6 &   8.56 &   404176 &  32544 &  629.04 \\
10 & 7 &   25.17 & 1247335 & 90589 &   93.40 & 6 &  22.02 &  1033821 &  74136 &  925.30 \\
11 & 8 &   85.42 & 3643870 &240258 &  518.61 & 7 &  64.59 &  3110693 & 203313 &  \multicolumn{1}{c||}{$\infty$}\\
12 & 8 &  234.39 & 8899673 &533226 &12343.21 & 7 & 185.27 &  7596239 & 451212 &  \multicolumn{1}{c||}{$\infty$} \\
\hline
\end{tabular}
  \caption{SAT-solving for $n$-channel, depth-$d$ sorting networks:
  each instance runs on a single thread,  \bee\  
  compile times and SAT-solving times are in seconds (timeout
  is $1$~week of computation).}
  \label{fig:satPlain}
\end{table}
For each $n$ (number of channels) with $5\leq n\leq 12$, there is a row in
the table that indicates: the depth of the network we seek ($d$ and
$d'$), the encoding time (using \bee), the size of
the resulting CNF (number of clauses and variables), and the SAT-solving
time. Times are indicated in seconds.
The left side of the table details satisfiable instances where we seek
a sorting network of optimal depth $d$. For $1\leq n\leq 10$, we take
for $d$ the known optimal depth, and for $n>10$ we take the best known
upper bound (see Table~\ref{tab:previousoptimal}).
The right side of the table details the (suspected to be)
unsatisfiable instances, where $d'$ is one less than the value~$d$.

Observe that, when $n=12$, in the search for a sorting network of
depth~$8$, the encoding creates a CNF with circa $8.9$~million clauses, and
a solution is found after about $3.5$~hours. On the other hand, for
$n=11$ the SAT solver is not able to prove unsatisfiability of
$\Psi(11,7,\BB^{11}_{un})$ even after one week of computation, and similarly for
$\Psi(12,7,\BB^{12}_{un})$. So this encoding suffices to prove depth optimality
of networks on up to $10$ channels. We now show how to extend it to handle more
channels.

\subsection{Optimal-depth sorting networks given a prefix}

Recall that, to show the existence of a
sorting network of a given size, it suffices to restrict attention to
networks with a fixed first layer
(Lemma~\ref{lem:parberry}). Furthermore (Corollary~\ref{lem:complete}), it
suffices to focus on second layers from  the set $R_n$.

We now show how to capitalize on these results. In general, let $n$ be
the number of channels and $d$ be the depth of a particular comparator network.
Then, given a
prefix $C$ consisting of layers $C=L_1;L_2;\cdots;L_{d'}$ we can
encode the property that the prefix of the network is $C$ by the
following formula:
\begin{equation}
\varphi^{\emph{fixed}}(n,d,C)
     = 
           \bigwedge_{\scriptsize\begin{array}{c}
             1 \leq \ell < d'\\
             1 \leq i<j \leq n
           \end{array}}
           c^\ell_{i,j}\leftrightarrow (i,j)\in L_\ell
\end{equation}
where the conjuncts fix the values of the Boolean variables
$c^\ell_{i,j}$ in $\networkVars^d_n$ to correspond to the positions
of the comparators in the given prefix $C$.

Given a set of Boolean inputs $X \subseteq \BB^n$, there is an
$n$-channel, depth-$d$ network with prefix $C$ sorting all inputs from
$X$ if and only if the following formula is satisfiable.
\begin{equation}
\label{satEncoding2}
  \Psi_C(n,d,X)
     = 
     \Psi(n,d,X) \wedge \varphi^{\emph{fixed}}(n,d,C)
\end{equation}
Note that Boolean sequences sorted after the application of $C$ remain
sorted regardless of the positioning of the comparators in the
subsequent layers. Thus, to show that there is a sorting network on
$n$ channels with depth~$d$ that begins with prefix $C$, it suffices to
restrict the set $X$ to Boolean sequences that are unsorted after application of
the prefix $C$.
Letting $\BB^n_{\mathit{un(C)}} \subseteq \BB^n$ denote the set of such sequences,
the following result holds.

\begin{lemma}
\label{lem:sat2}
There exists a sorting network on $n$ channels with depth $d$ and prefix $C$
if and only if the formula
$\Psi_C\left(n,d,\BB^n_{\mathit{un(C)}}\right)$ is satisfiable.
\end{lemma}

\subsubsection*{Experiment II}

According to Lemma~\ref{lem:parberry} we can fix the first level of
the network, and thus apply the encoding
$\Psi_C\left(n,d,\BB^n_{\mathit{un(C)}}\right)$ where $C$ consists of
a single maximal layer on $n$ channels. We take $C$ to be $F'_n =
\sset{(i,n-i+1)}{1\leq i\leq \left\lfloor\frac{n}{2}\right\rfloor}$.
Table~\ref{fig:satFix1} illustrates the results for the appropriate
instances of the formula 
$\Psi_C\left(n,d,\BB^n_{\mathit{un(F'_n)}}\right)$. Each instance (2 per
line in the table) is run on a single thread of the cluster. As in
Table~\ref{fig:satPlain}, the satisfiable instances are described on
the left and the unsatisfiable instances on the right.

\begin{table}[htb]
\centering\scriptsize\begin{tabular}{|l||r|r|r|r|r||r|r|r|r|r||}
\cline{2-11} \multicolumn{1}{c||}{}
    &\multicolumn{5}{c||}{optimal sorting networks (sat)}      
    &\multicolumn{5}{c||}{smaller networks (unsat)}\\
\hline
$n$ &$d$&
\multicolumn1{c|}{\bee} &
\multicolumn1{c|}{\#clauses} &
\multicolumn1{c|}{\#vars} &
\multicolumn1{c||}{SAT} &
$d'$&
\multicolumn1{c|}{\bee} &
\multicolumn1{c|}{\#clauses} &
\multicolumn1{c|}{\#vars} &
\multicolumn1{c||}{SAT} \\
\hline
5  & 5 &   0.04 &     1366 &    238 &     0.00  & 4  & 0.01&     915&   151&   0.00\\
6  & 5 &   0.09 &     3137 &    441 &     0.01  & 4  & 0.04&    2083&   276&   0.01\\
7  & 6 &   0.48 &    12699 &   1509 &     0.06  & 5  & 0.15&    9487&  1089&   0.03\\
8  & 6 &   0.97 &    26414 &   2682 &     0.17  & 5  & 0.35&   19680&  1930&   0.06\\
9  & 7 &   3.34 &    90846 &   8150 &     1.25  & 6  & 1.41&   72337&  6353&   1.00\\
10 & 7 &   6.23 &   177067 &  14091 &    17.15  & 6  & 2.81&  140847& 10978&   1.83\\
11 & 8 &  18.11 &   547708 &  39386 &   104.16  & 7  & 9.89&  454563& 32245& 282.04\\
12 & 8 &  38.84 &  1018902 &  66206 &   211.51  & 7  & 17.02&  845232& 54192& 521.62\\
13 & 9 & 112.81 &  2927622 & 174766 &  1669.88  & 8  & 49.26& 2500930&147902& \multicolumn{1}{c||}{$\infty$} \\
14 & 9 & 110.95 &  5264817 & 288609 & 56654.37  & 8  & 90.63& 4496413&244234& \multicolumn{1}{c||}{$\infty$} \\
\hline
\end{tabular}
  \caption{SAT-solving for $n$-channel, depth-$d$ sorting networks with
    $1$-layer filters: each instance runs on a single thread, \bee\ 
  compile times and SAT-solving times are in seconds (timeout
  is $1$~week of computation). }
  \label{fig:satFix1}
\end{table}

Note that the CNFs for this experiment are smaller than in the
previous experiment as fixing the first layer of the network reduces
the set of unsorted inputs considered to
$\BB^n_{\mathit{un(F'_n)}}$. 
This derives from the fact that,
as discussed before, sorted inputs (here, to
the second layer) can be ignored in the encoding.
For example, for $n=12$ we reduce the CNF sizes from $8.9$ and
$7.6$~million (see Table~\ref{fig:satPlain}) to $1.02$ and
$0.8$~million clauses, respectively. 
  
Observe that the encoding based on one layer filters suffices to prove
depth optimalily of networks for $n=11$ and $n=12$ channels in under $10$~minutes
with computational power equivalent to that of a standard desktop
computer.

\subsubsection*{Experiment III}

In our third experiment, we capitalize on Corollary~\ref{lem:complete},
which states that $R_n$ is a complete set of two layer filters on $n$
channels, i.e.~that there exists an $n$ channel sorting network of depth $d$
if and only if there exists one that extends one of the prefixes $C\in R_n$.

As an example, for $n=13$, $|R_{13}|=117$, and so, to determine whether
there exists a $13$-channel, depth-$8$ sorting network, it suffices to
determine whether any one of $117$ independent SAT instances
$\Psi_C\left(n,d,\BB^n_{\mathit{un(C)}}\right)$,
for $C\in R_{13}$,
is satisfiable.

\begin{table}[b]
\centering\scriptsize\begin{tabular}{|r|r|r||r|r|r|r|r||r|r||}
\cline{4-10} \multicolumn{3}{c||}{}
    &\multicolumn{5}{c||}{optimal depth: fastest satisfiable instance}      
    &\multicolumn{2}{c||}{total solving times}  \\
\hline
$n$ & $|R_n|$ & $d$& ins{.} &
\multicolumn1{c|}{\bee} &
\multicolumn1{c|}{\#clauses} &
\multicolumn1{c|}{\#vars} &
\multicolumn1{c||}{SAT} &
\multicolumn1{c|}{\bee} &
\multicolumn1{c||}{SAT} \\
\hline
5 & 4 & 5 & 3 & 0.01 & 534 & 97 & 0.00 & 0.03 & 0.01\\
6 & 5 & 5 & 5 & 0.01 & 1047 & 156 & 0.00 & 0.04 & 0.01\\
7 & 8 & 6 & 3 & 0.06 & 4537 & 569 & 0.01 & 0.49 & 0.08\\
8 & 12 & 6 & 10 & 0.07 & 6952 & 740 & 0.02 & 1.01 & 0.18\\
9 & 22 & 7 & 5 & 0.38 & 26019 & 2447 & 0.06 & 10.52 & 4.66\\
10 & 21 & 7 & 19 & 0.61 & 50573 & 4216 & 0.36 & 2.28 & 2.02\\
11 & 48 & 8 & 10 & 2.67 & 171357 & 13129 & 0.60 & 126.99 & 753.71\\
12 & 50 & 8 & 43 & 2.77 & 206776 & 14088 & 4.08 & 57.07 & 481.13\\
13 & 117 & 9 & 112 & 13.28 & 922363 & 56679 & 10.71 & 1711.99 & 38185.55\\
14 & 94 & 9 & 86 & 37.80 & 1124987 & 64318 & 123.13 & 6206.11 & $\infty$(43)\\
15 & 262 & 9 & 169 & 84.90 & 2684977 & 139181 & 19737.38 & 33176.39 & $\infty$(1)\\
16 & 211 & 9 & 188 & 116.36 & 3179978 & 155456 & 30509.58 & 46968.71 & $\infty$(1)\\
\hline
\end{tabular}
  \caption{SAT-solving for $n$-channel, depth-$d$ sorting networks with
    $2$-layer filters. These are the satisfiable instances (at least for
    one $C\in R_n$): fastest satisfiable
    instances detailed on the left; and total costs on the right: 
    \bee\ compile times and SAT-solving times are in
    seconds. Here, $\infty(k)$ indicates that $k$ 
    instances terminated within $24$~hours (each on a single core).}
  \label{fig:satFix2a}
\end{table}
In Table~\ref{fig:satFix2a} we illustrate results for the instances
with optimal depth $d$ for $1\leq n\leq 10$, and the best known upper
bound $d$ for $11\leq n\leq 16$. We consider the two-layer filters in
the sets $R_n$ as described in Section~\ref{sec:prefix}. For the row
corresponding to $n$ channels, we have $|R_n|$ instances, and each
instance is run on a single thread from the cluster.
The left part of the table describes the fastest satisfiable instance
from the $|R_n|$ instances. The instance number given in the third
column indicates a particular filter $C\in R_n$ (the instances are
detailed at
\url{http://www.cs.bgu.ac.il/~mcodish/Papers/Appendices/SortingNetworks/twoLayerFilters.pl}).
When running the instances with the full capacity of the cluster, we
can abort the computation as soon as the first satisfiable instance is
found. 
The right part of the table specifies the total cost of the
computation, and indicates the time required for each value of $n$ on a
single core. Here, we indicate the total compile times and SAT-solving
times for all $|R_n|$ instances. 
Each instance was limited to run for
$24$~hours on a single core, and $\infty(k)$ indicates that $k$ instances
terminated within $24$~hours (each on a single core).

For the unsatisfiable instances, we introduce one additional
optimization.  Consider again Equation~(\ref{satEncoding2}). A sorting network with prefix $C$ must sort all of its unsorted inputs $\BB^n_{\mathit{un(C)}}$. However, if we consider any specific subset
of $B\subseteq\BB^n_{\mathit{un(C)}} $ and show that there is no comparator network that
sorts the elements of $B$, then there is also no comparator
network that sorts all the unsorted inputs.

In particular, we consider length-$n$ Boolean sequences that have
sufficiently long prefixes of zeroes and suffixes of ones. Given an
integer $w<n$ and a set $B\subseteq\BB^n$, we denote the set $B{\upharpoonright}
w=\sset{b\in B}{b=0^{\ell_1}.\BB^{n-w}.1^{\ell_2}, \ell_1+\ell_2 = w}$, which
we refer to as the \emph{windows} of size $w$ of $B$.
\begin{table}
\centering\scriptsize\begin{tabular}{|r| r |r||r|r|r|r|r|r||r|r||}
\cline{4-11} \multicolumn{3}{c||}{}
    &\multicolumn{6}{c||}{$d'<$ optimal depth: slowest (unsat) instance}
    &\multicolumn{2}{c||}{total solving times}  \\
\hline
$n$& $|R_n|$ &$d'$& ins{.} &$w$ &
\multicolumn1{c|}{\bee} &
\multicolumn1{c|}{\#clauses} &
\multicolumn1{c|}{\#vars} &
\multicolumn1{c||}{SAT} &
\multicolumn1{c|}{\bee} &
\multicolumn1{c||}{SAT} \\
\hline
5 & 4 & 4 & 1 & 2 & 0.00 & 268 & 44 & 0.00 & 0.02 & 0.00\\
6 & 5 & 4 & 1 & 2 & 0.01 & 511 & 67 & 0.00 & 0.02 & 1.00\\
7 & 8 & 5 & 6 & 2 & 0.04 & 2965 & 348 & 0.01 & 0.30 & 1.04\\
8 & 12 & 5 & 4 & 3 & 0.09 & 4423 & 458 & 0.01 & 0.80 & 4.09\\
9 & 22 & 6 & 16 & 3 & 0.20 & 14716 & 1416 & 0.05 & 4.16 & 0.75\\
10 & 21 & 6 & 1 & 4 & 0.41 & 20027 & 1815 & 0.09 & 6.35 & 5.08\\
11 & 48 & 7 & 46 & 4 & 0.73 & 52365 & 4314 & 1.41 & 53.79 & 52.26\\
12 & 50 & 7 & 21 & 5 & 1.02 & 62051 & 4826 & 2.39 & 93.79 & 72.56\\
13 & 117 & 8 & 77 & 3 & 14.68 & 464035 & 29958 & 749.27 & 1047.93 & 35726.26\\
14 & 94 & 8 & 16 & 4 & 20.27 & 448903 & 27473 & 2627.60 & 2342.96 & 81533.49\\
15 & 262 & 8 & 189 & 7 & 8.32 & 278312 & 18217 & 746.42 & 3491.06 & 127062.36\\
16 & 211 & 8 & 112 & 7 & 22.85 & 453810 & 27007 & 1756.29 & 4448.66 & 152434.55\\
\hline
\end{tabular}
  \caption{SAT-solving for $n$-channel, depth-$d'$ sorting networks with
    $2$-layer filters. These are the unsatisfiable instances (for all
    $C\in R_n$): slowest
    instances detailed on the left; and total costs on the right:
    \bee\ compile times and SAT-solving times are in
    seconds; and $w$ is window size.}
  \label{fig:satFix2b}
\end{table}

Table~\ref{fig:satFix2b} depicts results for the instances with the set of inputs equal to $\BB^n_{\mathit{un(C)}}{\upharpoonright} w$ and depth
$d'=d-1$ where $d$ is the known optimal depth for $1\leq n\leq 10$,
and the best known upper bound $d$ for $11\leq n\leq 16$.  We consider
the two-layer filters in the sets $R_n$ as described in
Section~\ref{sec:prefix}.  For the row corresponding to $n$ channels,
we again have $|R_n|$ instances, and each instance is run on a single thread
from the cluster.  The left part of the table describes the slowest
unsatisfiable instance from the $|R_n|$ instances, including the
largest window size $w$ for which unsatisfiability is obtained. The
table also specifies, in the fourth column, the index of the slowest
instance.
For the unsatisfiable instances, we need to run all instances to
determine that all are unsatisfiable, and the parallel cost is the time
of the slowest instance (in the left part of the table).
The right part of the table specifies the total cost of the
computation, indicating the time required for each value of $n$ on a
single core. Here we indicate the total compile times and SAT-solving
times for all $|R_n|$ instances.

While we used multiple threads on a cluster for our experiments, the two instances relevant for our results, $n = 11$ and $n = 13$, could be run on a single thread on a desktop computer in $2$ minutes and $10$ hours, respectively.
Once we show that $T(11)=8$, it follows from the known bounds
(Table~\ref{tab:previousoptimal}) that also $T(12)=8$, because $T$ is monotonic.
Likewise, once we show that $T(13)=9$, it follows that also $T(14)=T(15)=T(16)=9$.
The motivation for also computing them was to show that this approach actually scales up to $n=16$.

\section{Conclusions and future work}
\label{sec:concl}

We have shown that $T(11) = T(12) = 8$ and $T(13) = T(14) = T(15) = T(16) =
9$, i.e., we have proven that the previously known upper bounds on the
optimal depth of $n$-channel sorting networks are tight for $11 \leq n
\leq 16$.  This closes the six smallest open instances of the optimal-depth
sorting network problem, thereby proving depth optimality of the sorting
networks for $n \leq 16$ given in \cite{Knuth73} more than four
decades ago.

The next smallest open instance of the optimal-depth sorting network
problem is for $n=17$ where the best known upper bound is $11$.
Attempting to show that there is no sorting network of depth $10$
requires analyzing the SAT encodings given the networks in
$R_{17}$. The resulting $609$ formulas have more than five million
clauses each, and none could be solved within a couple of weeks.  It
appears that establishing the optimal depth of sorting networks with more
than $16$ channels is a hard challenge that will require prefixes with
more than $2$ layers.

The encoding into SAT that we propose in this paper is of size exponential in the number of channels,~$n$. This is also the case for the encoding
presented in~\cite{MorgensternS11}. The encoding is of the form $\exists\forall\varphi$ (does there \emph{exist} a network
that sorts \emph{all} of its inputs), and is easily shown to be in
$\Sigma_2^P$.  We expect that, similar to the problem of circuit
minimization~\cite{DBLP:conf/focs/Umans98}, it is also complete in
$\Sigma_2^P$, although we have not succeeded to prove this. We do not
expect that there exists a polynomial size encoding to SAT.

\section*{Acknowledgements}

We thank Donald E.~Knuth for suggesting the idea of using reflections to reduce
the number of representative two-layer networks.


\end{document}